\documentclass[conference,letterpaper]{IEEEtran}
\newcommand{\appref}[2]{Appendix~\ref{#2}}
\usepackage[utf8]{inputenc}
\usepackage[T1]{fontenc}
\usepackage{url}
\usepackage{ifthen}
\IEEEoverridecommandlockouts
\usepackage[cmex10]{amsmath}
\usepackage{amssymb,amsfonts}
\usepackage{algorithm}
\usepackage{graphicx}
\usepackage[caption=false,font=footnotesize]{subfig}
\usepackage{textcomp}
\usepackage{needspace}
\usepackage{algpseudocode}
\usepackage{dsfont}
\usepackage{changepage}
\usepackage{setspace}
\usepackage{pdfpages}
\usepackage{mathrsfs}
\usepackage{mathtools}
\usepackage{amsthm}
\usepackage{diagbox}
\usepackage{multirow}
\usepackage{lscape}
\usepackage{rotating}
\usepackage{cite}

\newtheorem{theorem}{Theorem}

\newtheorem{Definition}{Definition}

\newtheorem{lemma}{Lemma}
\theoremstyle{definition}
\usepackage{comment}
\usepackage{booktabs}
\usepackage[hidelinks]{hyperref}

\allowdisplaybreaks[4]

\def\BibTeX{{\rm B\kern-.05em{\sc i\kern-.025em b}\kern-.08em
    T\kern-.1667em\lower.7ex\hbox{E}\kern-.125emX}}


\begin{document}

\title{Joint Source--Channel Coding for ISAC with and without Source--State Correlation}

\author{\IEEEauthorblockN{Gefei Peng and Youlong Wu}
    \IEEEauthorblockA{
    ShanghaiTech University, Shanghai, China}
    \{penggf2025, wuyl1\}@shanghaitech.edu.cn
}

\maketitle
\begin{abstract}
    Integrated Sensing and Communication (ISAC) systems have garnered significant attention due to their capability to simultaneously achieve efficient communication and environmental sensing. A fundamental problem is to characterize the performance tradeoff between sensing and communication. In this paper, we investigate a joint source--channel coding (JSCC) problem in an ISAC system, where the transmitter observes a source sequence and aims to communicate it to the receiver under a given distortion criterion over a memoryless state-dependent channel, while simultaneously estimating the channel state based on feedback signals. We first prove that when the source and channel state are independent, separate source and channel coding (SSCC) is optimal for general source-to-channel blocklength ratios. For correlated source--state pairs with equal source and channel blocklengths, we derive single-letter inner and outer bounds on the achievable region. We further exhibit an example in which joint coding strictly outperforms the SSCC benchmark, and the proposed inner and outer bounds match at the specified sensing-distortion and input-cost budgets.

\end{abstract} 


\section{Introduction}
Integrated sensing and communications (ISAC) jointly consider communication and sensing within a common wireless system, thereby improving resource utilization and reducing infrastructure overhead \cite{liu2022integrated, liu2022survey}. A fundamental problem of ISAC is to characterize the tradeoff between communication and sensing performance.

The capacity-distortion framework for ISAC over state-dependent memoryless channels (SDMCs) with feedback was introduced by Kobayashi \textit{et al.} \cite{kobayashi2018joint}, where the single-input single-output trade-off boundary was evaluated via the Blahut-Arimoto algorithm. This formulation was subsequently extended to memoryless multiple access channels \cite{kobayashi2019joint}, state-dependent broadcast channels \cite{ahmadipour2021joint, ahmadipour2022information}, bistatic ISAC setups with physically separated sensing nodes \cite{liu2024bistatic, jiao2025information}, and physical-layer secrecy constraints \cite{gong2024secrecy}.

Despite these advances, existing information-theoretic ISAC studies have primarily focused on channel coding with uniformly distributed messages independent of the channel state sequence. In many practical ISAC scenarios, however, the source sequence to be transmitted and the underlying channel state (e.g., target or environmental properties) are inherently correlated, or the source--channel blocklengths may be mismatched. Classical studies have characterized optimal source--channel matching \cite{gastpar2006code} and revisited separation for general sources and channels \cite{vembu1995source}; recent progress in joint source--channel coding (JSCC) is surveyed in \cite{gunduz2024joint}. Yet, the fundamental interplay between JSCC and ISAC remains largely unexamined. In ISAC systems, the channel state both affects the channel transition law and is the quantity to be estimated at the transmitter through generalized feedback. This unique coupling raises two critical open questions: 1) \textit{Does SSCC optimality hold in ISAC systems when the source and state are independent, particularly under blocklength mismatches?} 2) \textit{When the source and channel state are correlated, how can we characterize the achievable joint source-reconstruction and state-estimation performance via JSCC?}

To answer these fundamental questions, this paper investigates the JSCC problem for state-dependent memoryless ISAC systems with feedback. Our main contributions are summarized as follows:
\begin{itemize}
    \item We prove that when the source and channel state are statistically independent, SSCC optimality holds for ISAC across general source-to-channel blocklength ratios. This establishes that SSCC is optimal among all admissible JSCC schemes.
    \item For correlated source--state pairs with equal source and channel blocklengths, we develop a joint source--channel coding scheme and derive single-letter inner and outer bounds for the corresponding JSCC problem. Through an example with correlated source and state, we further show that the proposed joint scheme can strictly outperform the SSCC benchmark.

\end{itemize}

 \emph{Paper organization and notations:} Section~\ref{sec:system_model} introduces the system model; Section~\ref{sec:separated_benchmark} presents the SSCC benchmark; Section~\ref{sec:main_results} gives the main results; Section~\ref{sec:numerical_examples} presents a numerical example; and Section~\ref{sec:conclusion} concludes the paper. Calligraphic fonts denote sets (e.g., $\mathcal{X}$), whereas $\mathbb{R}$ and $\mathbb{R}_0^+$ indicate the sets of real and non-negative real numbers. Random variables are uppercase letters like $X$, with realizations in lowercase such as $x$; $[1:X]$ denotes $\{1,\dots,X\}$, and $X^n$ denotes the tuple $(X_1,\dots,X_n)$, and $X^{-i}:=(X^{i-1},X_{i+1}^n)$ omits its $i$th component. We use i.i.d. for independent and identically distributed, $P_X(x)=\Pr[X=x]$, and base-2 logarithms.

\section{System Model}
\label{sec:system_model}

Consider the JSCC model depicted in Fig.~\ref{fig:sys}, where the transmitter observes a source sequence $U^k$ and conveys it to the receiver over an SDMC using $n$ channel uses. The receiver reconstructs $U^k$ subject to a source-distortion constraint $D_u\ge0$, while the transmitter estimates the channel state from generalized feedback subject to a sensing-distortion constraint $D_s\ge0$.

\begin{figure}[t]
    \centering
    \includegraphics[width=\columnwidth]{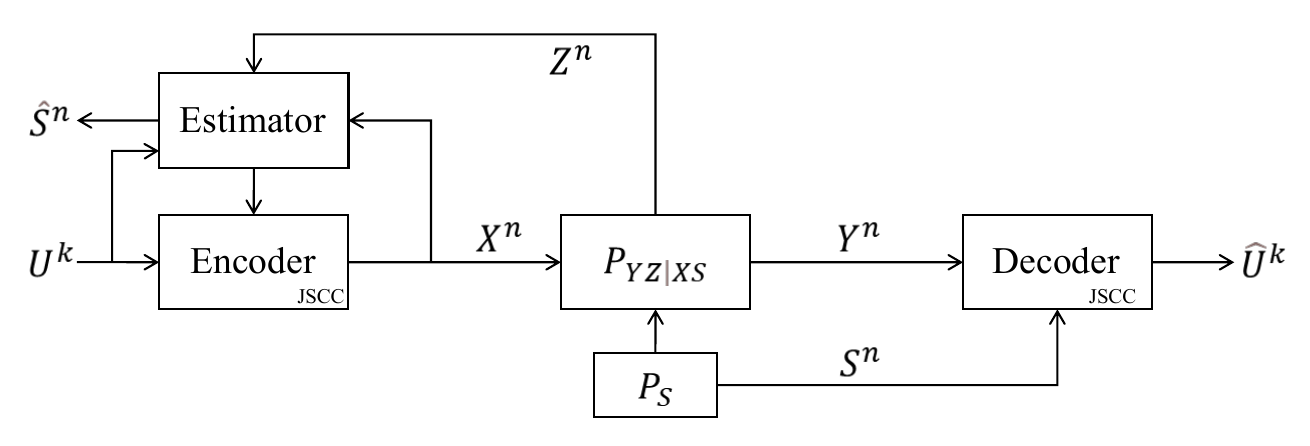}
    \caption{The JSCC framework for the ISAC problem.}
    \label{fig:sys}
\end{figure}

The source sequence $U^k$ and the channel-state sequence $S^n$ have i.i.d. marginals $P_U$ and $P_S$, respectively, with their dependence specified separately for Cases~1 and~2.
At channel use $i$, the transmitter has access to $(U^k,Z^{i-1})$
and sends $X_i$. Given $(X_i,S_i)=(x_i,s_i)$, the SDMC generates
$(Y_i,Z_i)$ according to the fixed transition law
$P_{YZ|XS}(\cdot,\cdot|x_i,s_i)$, independently of past channel
inputs, outputs, states, and feedback signals.
The receiver observes $(Y^n,S^n)$ for source reconstruction,
while the generalized feedback is available at the transmitter
for state estimation. All alphabets
$\mathcal U$, $\mathcal S$, $\mathcal X$, $\mathcal Y$,
$\mathcal Z$, $\hat{\mathcal U}$, and $\hat{\mathcal S}$
are finite.

Let $\gamma>0$ denote the asymptotic source-to-channel blocklength ratio, with $k/n\to\gamma$ as $n\to\infty$. We consider two cases: \begin{itemize}
\item \textbf{Case~1} considers a general source-to-channel blocklength ratio $\gamma > 0$ with \emph{statistically independent} source and state sequences, i.e., $P_{U^kS^n}(u^k,s^n) = P_{U^k}(u^k) P_{S^n}(s^n)$.
\item \textbf{Case~2} models element-wise memoryless correlation, where $k = n$ (i.e., $\gamma = 1$), and the pair sequence $(U^n,S^n)$ is jointly i.i.d. according to $P_{US}$, such that $P_{U^nS^n}(u^n,s^n) = \prod_{i=1}^n P_{US}(u_i,s_i)$.

\end{itemize}

A $(k,n)$ joint source--channel code consists of: i) an encoder that assigns the channel input $x_i(u^k,z^{i-1})\in\mathcal{X}$, $i=1,\ldots,n$, to each $u^k\in\mathcal{U}^k$; ii) a decoder at the communication receiver that produces an estimate $\hat{u}^k(y^n, s^n) \in \hat{\mathcal{U}}^k$; iii) a state estimator at the transmitter that produces an estimated state sequence $\hat{s}^n(x^n, z^n) \in {\hat{\mathcal{S}}^n}$ upon observing {$(x^n,z^n)$ in Case~1}. {In Case~2, it produces $\hat s^n(u^n,x^n,z^n)\in\hat{\mathcal S}^n$ upon observing $(u^n,x^n,z^n)$.}


The sensing and communication performance is measured by the expected average per-block distortion
\begin{IEEEeqnarray}{rCl}
   \Delta_s^{(n)} &:= \mathbb{E}[{d_s}(S^n, \hat{S}^n)] = \frac{1}{n}\sum_{i=1}^n \mathbb{E}[{d_s}(S_i, \hat{S}_i)], \label{eq:3.1a} \\
    \Delta_u^{(k)} &:= \mathbb{E}[{d_u}(U^k, \hat{U}^k)] = \frac{1}{k}\sum_{i=1}^k \mathbb{E}[{d_u}(U_i, \hat{U}_i)], \label{eq:3.1b}
\end{IEEEeqnarray}
where ${d_s}:\mathcal{S}\times\hat{\mathcal{S}}\mapsto\mathbb{R}_0^+$ and ${d_u}:\mathcal{U}\times\hat{\mathcal{U}}\mapsto\mathbb{R}_0^+$ are given bounded distortion functions, and the input cost is $\mathbb{E}[b(X^n)]=\frac{1}{n}\sum_{i=1}^n\mathbb{E}[b(X_i)]$ for a given cost function \(b:{\mathcal X}\mapsto\mathbb{R}_0^+\).

\begin{Definition}
\label{def:achievable}
    A tuple $(\gamma,D_u,D_s,B)$ is achievable if there exists
    a sequence of $(k,n)$ codes with
    $k/n\to\gamma$ as $n\to\infty$ that simultaneously satisfy
    \begin{subequations}\label{eq:achievable_conditions}
        \begin{align}
         &   {\limsup_{n \to \infty}}\, \Delta_u^{(k)}  \leq D_u \label{eq:data_dist},\\
         & {\limsup_{n \to \infty}}\, \Delta_s^{(n)}  \leq D_s \label{eq:state_dist}, \\
        &    {\limsup_{n \to \infty}}\, \frac{1}{n} \sum_{i=1}^n \mathbb{E}[b(X_i)]  \leq B. \label{eq:cost_constraint}
        \end{align}
    \end{subequations}
\end{Definition}

Our goal is to characterize the achievable rate--distortion--cost
tuples $(\gamma,D_u,D_s,B)$ and determine whether SSCC attains the same region as JSCC.

\section{SSCC Benchmark}
\label{sec:separated_benchmark}

To assess the fundamental impact of source--channel coupling in the
considered JSCC-ISAC setting, we first establish an SSCC benchmark.
In the correlated source--state setting, we adopt a conventional
block-delayed implementation: a description of the preceding source
block is transmitted over the current source--state block, which is
independent of the preceding block. Consequently, the channel-coding
message is independent of the current source--state block.
This benchmark does not encompass every modular source--channel
architecture whose channel inputs may remain correlated with the source through the source-coding index.
It provides a common reference for establishing SSCC optimality in
the independent setting and for quantifying the gains attainable
through JSCC over this benchmark in the correlated setting.

\begin{figure}[t]
    \centering
    \captionsetup[subfloat]{
        farskip=2pt,
        nearskip=0pt,
        captionskip=2pt
    }

    \subfloat[Separated source-coding component.]{%
        \includegraphics[width=\columnwidth]{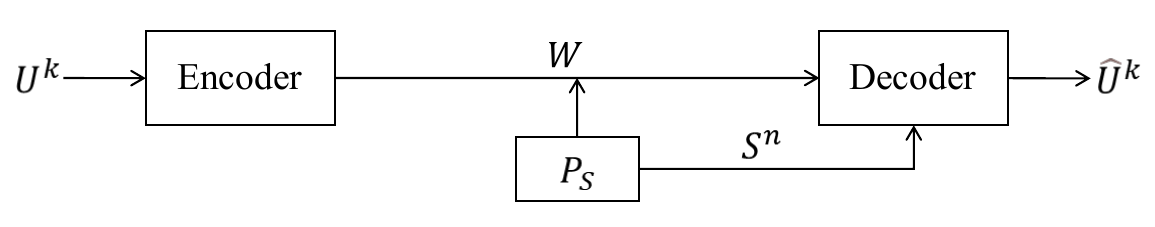}%
        \label{fig:separation_source}}

    \par\vspace{4pt}

    \subfloat[Separated channel-coding component.]{%
        \includegraphics[width=\columnwidth]{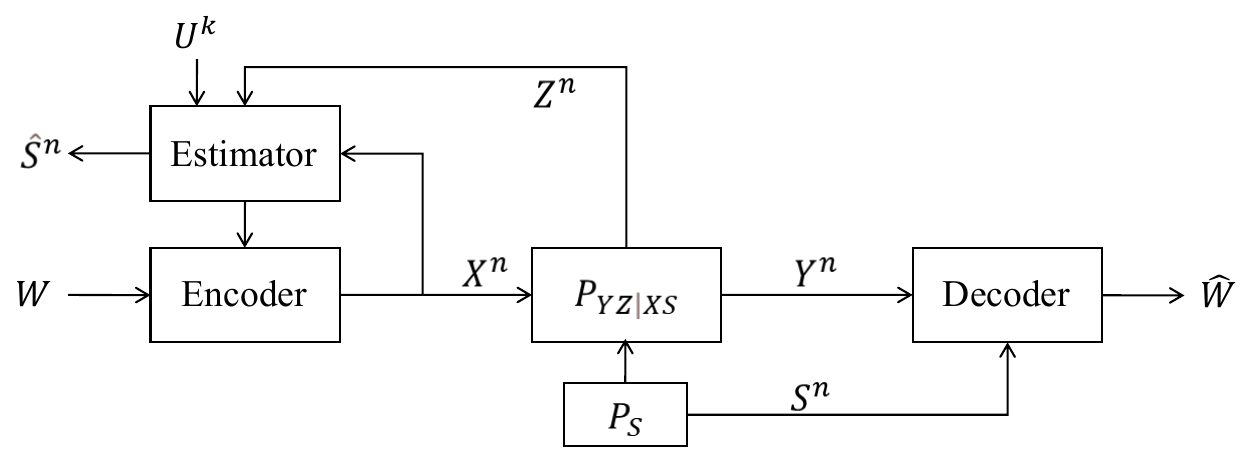}%
        \label{fig:separation_channel}}

    \caption{Components of the SSCC benchmark.}
    \label{fig:separation}
\end{figure}

\subsection{Source-Coding Component}
\label{sec:benchmark_source}
The source-coding component in Fig.~\ref{fig:separation_source}
is characterized by the following rate--distortion functions.

Define the source-reconstruction constraint set
\begin{equation*}
\mathcal{P}_{D_u}
=
\Bigl\{
P_{\hat U|U}:
\sum_{u,\hat u}
P_U(u)P_{\hat U|U}(\hat u|u)d_u(u,\hat u)
\le D_u
\Bigr\},
\end{equation*}
and let
\[
D_{u_{\min}}
:=
\min_{P_{\hat U|U}}
\sum_{u,\hat u}
P_U(u)P_{\hat U|U}(\hat u|u)d_u(u,\hat u).
\]

In Case~1, $U^k$ and $S^n$ are independent, so the receiver
side information $S^n$ does not reduce the source-coding rate.
The source-coding component therefore reduces to the classical
rate-distortion problem \cite[Ch.~10]{cover1999elements}, with
\begin{equation}
R(D_u)
:=
\min_{P_{\hat U|U}\in\mathcal P_{D_u}}
I(U;\hat U),
\label{eq:ratedisto_case1}
\end{equation}
where the joint law is $P_UP_{\hat U|U}$.

In Case~2, source coding with decoder side information $S^n$
is the classical Wyner--Ziv problem
\cite{wyner1976rate}; see also \cite[Ch.~11]{el2011network}.
Its rate-distortion function is
\begin{equation}
R_{U|S}(D_u)
:=
\min_{\substack{P_{V|U},\,\phi:\\
V-U-S,\;
P_{\hat U|U}\in\mathcal P_{D_u}}}
I(U;V|S),
\label{eq:c2_wz}
\end{equation}
where $P_{USV}=P_{US}P_{V|U}$, $\hat U=\phi(V,S)$,
and $P_{\hat U|U}$ is induced by $(P_{V|U},\phi)$; it suffices
to take $|\mathcal V|\le|\mathcal U|+1$.

\subsection{Channel-Coding Component}
\label{sec:benchmark_channel}

Ahmadipour et al.~\cite{ahmadipour2022information}
characterized the capacity--distortion--cost tradeoff for
a point-to-point state-dependent channel with generalized
feedback and receiver state information. In their formulation,
state estimation at the transmitter is based on $(X^n,Z^n)$,
as illustrated in
Fig.~\ref{fig:separation}\subref{fig:separation_channel}.

In the channel-coding component of the SSCC benchmark, the source-coding
index $W$ is conveyed via $X_i=\phi_i(W,Z^{i-1})$ and decoded from
$(Y^n,S^n)$.
In Case~1, this specialization recovers the classical
capacity--distortion--cost characterization. In Case~2,
source--state correlation modifies the sensing side of the
problem, since the source observation available at the
transmitter can provide additional information about the state.
The following general state-estimation lemma characterizes the resulting estimation rule
and yields the sensing costs entering the corresponding
channel-coding expressions.

\begin{lemma}
\label{Lemma1}
Fix a code under either source--state model of
Section~\ref{sec:system_model}, with strictly causal encoding
$X_i=x_i(U^k,Z^{i-1})$. On every positive-probability
conditioning event,
\begin{equation}
\begin{aligned}
&P_{S_i|U^kX_iZ_i}(s|u^k,x,z)\\
&\quad=
\frac{
P_{S_i|U^k}(s|u^k)P_{Z|SX}(z|s,x)}
{\sum_{\tilde s\in\mathcal S}
P_{S_i|U^k}(\tilde s|u^k)
P_{Z|SX}(z|\tilde s,x)}.
\end{aligned}
\label{eq:prob}
\end{equation}
Hence, a distortion-minimizing state estimate is
\begin{equation}
\hat s_i^*(u^k,x,z)
:=
\arg\min_{s'\in\hat{\mathcal S}}
\sum_{s\in\mathcal S}
P_{S_i|U^kX_iZ_i}(s|u^k,x,z)d_s(s,s'),
\label{eq:sensing_estimator}
\end{equation}
with ties resolved arbitrarily, and coordinatewise application
of these estimators minimizes $\Delta_s^{(n)}$.
In Case~1, this rule reduces to $\hat s^*(x,z)$; in Case~2,
it reduces to $\hat s^*(u_i,x,z)$.
\end{lemma}

\begin{proof}
See \appref{A}{app:state_estimation}.
\end{proof}

For Case~1, Lemma~\ref{Lemma1} yields the sensing cost
$c(x):=\mathbb E\!\left[d_s(S,\hat s^*(x,Z))\mid X=x\right]$.
Define
\begin{align}
\mathcal P_B
&:=\left\{P_X:\sum_x P_X(x)b(x)\le B\right\},
\nonumber\\
\mathcal P_{D_s}^{(1)}
&:=\left\{P_X:\sum_x P_X(x)c(x)\le D_s\right\}.
\label{region_case1}
\end{align}
For $B\ge0$ with $\mathcal P_B\ne\varnothing$, define
$D_{s,\min}^{(1)}(B):=
\min_{P_X\in\mathcal P_B}\sum_xP_X(x)c(x)$.
Then, for $D_s\ge D_{s,\min}^{(1)}(B)$,
\cite[Theorem~1]{ahmadipour2022information} gives
\begin{IEEEeqnarray}{rCl}
C_{\rm sep}(D_s,B)
=
C_{\inf}(D_s,B)
=
\!\!\!\max_{P_X\in\mathcal P_B\cap\mathcal P_{D_s}^{(1)}}\!\!
I(X;Y|S),\quad
\label{eq:channelcapacity_case1}
\end{IEEEeqnarray}
where $(X,S,Y,Z)\sim P_XP_SP_{YZ|SX}$.

For Case~2, Lemma~\ref{Lemma1} yields the source-dependent
sensing cost
\begin{IEEEeqnarray}{rCl}
c(u,x)
\!:=\!
\sum_{s,z}
&P_{S|U}(s|u)P_{Z|SX}(z|s,x)d_s\bigl(s,\hat s^*(u,x,z)\bigr).\quad
\label{eq:c2_sensing_cost}
\end{IEEEeqnarray}
For the Case~2 SSCC benchmark, the block arrangement and
strictly causal encoding imply that each $X_i$ is independent
of the current $(U_i,S_i)$. Let $P_X$ denote its pmf. Define
\begin{equation}
\mathcal P_{D_s}^{(2)}
:=
\left\{P_X:
\sum_{u,x}P_U(u)P_X(x)c(u,x)\le D_s
\right\}.
\label{eq:c2_sep_sensing_set}
\end{equation}
For $B\ge0$ with $\mathcal P_B\ne\varnothing$, define
$D_{s,\min}^{(2)}(B):=
\min_{P_X\in\mathcal P_B}
\sum_{u,x}P_U(u)P_X(x)c(u,x)$.
Then, for $D_s\ge D_{s,\min}^{(2)}(B)$,
\begin{equation}
C_{\rm sep}(D_s,B)
=
\max_{P_X\in\mathcal P_B\cap\mathcal P_{D_s}^{(2)}}
I(X;Y|S),
\label{eq:c2_csep}
\end{equation}
under $P_{USXYZ}=P_{US}P_XP_{YZ|SX}$; see
\appref{C}{app:c2_sep}.
When $U$ and $S$ are independent,
$c(u,x)=c(x)$ and
$\mathcal P_{D_s}^{(2)}=\mathcal P_{D_s}^{(1)}$, so
\eqref{eq:c2_csep} reduces to
\eqref{eq:channelcapacity_case1}.

\section{Main Results}
\label{sec:main_results}

\subsection{Independent Source and State: SSCC Optimality}

We first show that the SSCC benchmark
of Section~\ref{sec:separated_benchmark} is optimal when the
source and channel state are independent.

\begin{theorem}\label{main_case1}
    Consider the JSCC framework
    of Section~\ref{sec:system_model}
    with $P_{U,S}=P_UP_S$. For any $(D_u,D_s)$ with {\(D_u\ge D_{u_{\min}}\)}, \(D_s\ge D_{s,\min}^{(1)}(B)\), and \(B\ge0\) with $\mathcal P_B\ne\varnothing$, the following statements hold:
    \begin{enumerate}
        \item If \( \gamma R(D_u) < C_{\inf}(D_s,B) \), then \( (\gamma,D_u,D_s,B) \) is achievable.
        \item If \( (\gamma,D_u,D_s,B) \) is achievable, then \( \gamma R(D_u) \leq C_{\inf}(D_s,B) \).
    \end{enumerate}
\end{theorem}

\begin{proof}
We briefly outline the proof.
The converse combines the rate--distortion lower bound with a
channel mutual-information upper bound.
Achievability follows from SSCC.
The complete proof, including the treatment of boundary points,
is given in \appref{B}{app:main_case1}.
\end{proof}

Therefore, the SSCC benchmark is optimal
in Case~1 and attains the same achievable rate--distortion--cost
region as general JSCC.

\subsection{Correlated Source and State: Joint Coding}

Throughout this subsection, we consider Case~2 of the JSCC
framework in Section~\ref{sec:system_model}, where $k=n$ and
$(U_i,S_i)_{i=1}^n$ are i.i.d. according to $P_{US}$.
We first present two achievable conditions and subsequently
derive a converse bound.

By Lemma~\ref{Lemma1}, every admissible JSCC code satisfies
\begin{equation}
\frac1n\sum_{i=1}^n
\mathbb E[c(U_i,X_i)]
\le
\Delta_s^{(n)},
\label{eq:c2_sensing_lower}
\end{equation}
with equality attained by the optimal symbol-by-symbol estimator.
This motivates the following input-cost and sensing-distortion
constraint sets for source-dependent channel inputs:
\begin{IEEEeqnarray}{rCl}
\mathcal Q_B
&:=
\left\{
P_{X|U}:
\sum_{u,x}P_U(u)P_{X|U}(x|u)b(x)\le B
\right\},
\nonumber\\
\mathcal Q_{D_s}
&:=
\left\{
P_{X|U}:
\sum_{u,x}P_U(u)P_{X|U}(x|u)c(u,x)\le D_s
\right\}\!\!.\quad
\label{eq:c2_joint_constraints}
\end{IEEEeqnarray}
The SSCC benchmark is obtained as the specialization
$P_{X|U}=P_X$, in which case
$P_X\in\mathcal P_B\cap\mathcal P_{D_s}^{(2)}$.

We first construct an achievable scheme that exploits
source--state correlation through source-dependent channel inputs.
The auxiliary random variable $A$ serves as the Cover--Chiang
channel-coding auxiliary~\cite{cover2002duality} for conveying the preceding block's
Wyner--Ziv description while exploiting the current source block
as encoder side information.

\begin{theorem}[JSCC achievability]
\label{thm:c2_wz_cc}
Consider Case~2 of the JSCC framework in
Section~\ref{sec:system_model}.
For any $D_u,D_s,B\ge0$, if there exists a finite-alphabet
auxiliary random variable $A$ and a conditional pmf
$P_{A,X|U}$ with
$P_{X|U}\in\mathcal Q_B\cap\mathcal Q_{D_s}$ such that
\begin{equation}
R_{U|S}(D_u)
\le
I(A;Y,S)-I(A;U),
\label{eq:c2_wz_cc_condition}
\end{equation}
then $(1,D_u,D_s,B)$ is achievable, where the mutual
informations are evaluated under
$P_{USAXYZ}=P_{US}P_{A,X|U}P_{YZ|SX}$.

Moreover, for every
$P_X\in\mathcal P_B\cap\mathcal P_{D_s}^{(2)}$,
choosing $A=X$ independent of $U$ satisfies the same
input-cost and sensing-distortion constraints and gives
\begin{equation}
I(A;Y,S)-I(A;U)
=
I(X;Y|S).
\label{eq:c2_cc_vs_sep}
\end{equation}
Thus, the resulting achievable region contains the
SSCC benchmark's achievable region under
\[
R_{U|S}(D_u)<C_{\rm sep}(D_s,B).
\]
\end{theorem}

\begin{proof}
We briefly outline the coding scheme.
Fix a conditional pmf $P_{A,X|U}$ for which
\eqref{eq:c2_wz_cc_condition} holds strictly.
Let $R$ denote the Wyner--Ziv description rate and $R_L$
an auxiliary covering rate.
In each length-$m$ block, generate
$2^{m(R+R_L)}$ i.i.d. $A^m$ codewords in
$2^{mR}$ equal bins.
For $R>R_{U|S}(D_u)$, the preceding Wyner--Ziv index,
which is independent of the current $(U^m,S^m)$, selects a bin.
Within that bin, select a codeword jointly typical with $U^m$
and use the corresponding input sequence generated according
to $P_{X|U,A}$.
The receiver decodes the unique codeword jointly typical with
$(Y^m,S^m)$.

The covering and packing conditions are
$R_L>I(A;U)$ and $R+R_L<I(A;Y,S)$.
The recovered bin index and the decoder side information
corresponding to the preceding source block yield its
Wyner--Ziv reconstruction.
Joint typicality and the estimator $\hat s^*(U,X,Z)$ ensure
the sensing-distortion and input-cost constraints.
The complete proof is given in
\appref{D}{app:c2_cc}.
\end{proof}

We next enlarge the achievable region of
Theorem~\ref{thm:c2_wz_cc} by exploiting additional decoder
side information and generalized feedback.
The supplementary description $\widetilde V$ is formed from
$(U,A,X,Z)$ and decoded with side information $(A,Y,S)$.


\begin{theorem}[Feedback-assisted JSCC achievability]
\label{thm:c2_inner}
Consider Case~2 of the JSCC framework in
Section~\ref{sec:system_model}.
For any $D_u,D_s,B\ge0$, if there exist finite-alphabet
auxiliary random variables $A$ and $\widetilde V$ and
conditional pmfs $P_{A,X|U}$ and
$P_{\widetilde V|U,A,X,Z}$ such that
\begin{equation}
I(U,X,Z;\widetilde V|A,Y,S)
\le
I(A;Y,S)-I(A;U),
\label{eq:c2_inner_condition}
\end{equation}
then $(1,D_u,D_s,B)$ is achievable, where
$P_{USAXYZ\widetilde V}
=P_{US}P_{A,X|U}P_{YZ|SX}
P_{\widetilde V|U,A,X,Z}$,
$P_{X|U}\in\mathcal Q_B\cap\mathcal Q_{D_s}$, and
$P_{\hat U|U}\in\mathcal P_{D_u}$ for some reconstruction
$\hat U=\phi(A,\widetilde V,Y,S)$.

Moreover, the achievable region in
Theorem~\ref{thm:c2_wz_cc} is contained in this region.
\end{theorem}

Unlike the Wyner--Ziv source-coding auxiliary $V$,
$\widetilde V$ is not constrained to satisfy the Markov chain
$\widetilde V-U-S$.

\begin{proof}
We briefly outline the coding scheme.
For blocks of length $m$, the receiver decodes the current
$A^m$ codeword.
After each block, the encoder forms a supplementary description
$\widetilde V^m$ from $(U^m,A^m,X^m,Z^m)$ and conveys its bin
index in the next block.
Since this index is independent of the next source--state block,
the Cover--Chiang construction can be used for its transmission.
The receiver recovers $\widetilde V^m$ from the decoded bin index
and $(A^m,Y^m,S^m)$ of the preceding block.

When \eqref{eq:c2_inner_condition} holds strictly, the covering
and packing conditions can be satisfied simultaneously.
Standard typicality arguments, together with the reconstruction
rule and Lemma~\ref{Lemma1}, establish the source-distortion,
sensing-distortion, and input-cost constraints.
The complete proof is given in \appref{E}{app:c2_inner}.
\end{proof}

Theorems~\ref{thm:c2_wz_cc} and
\ref{thm:c2_inner} provide sufficient conditions for Case~2.
To assess their tightness, we next derive a converse condition
satisfied by every achievable tuple.
The auxiliary $\overline V$ introduced below is used for
single-letterization in the converse and has no operational
coding role.

\begin{theorem}[Converse]
\label{thm:c2_outer}
Consider Case~2 of the JSCC framework in
Section~\ref{sec:system_model}.
If $(1,D_u,D_s,B)$ is achievable, then there exist a
conditional pmf $P_{X|U}$ and a finite-alphabet auxiliary
random variable $\overline V$ with conditional pmf
$P_{\overline V|U,X,Z}$ such that
\begin{align}
I(U;\overline V|Y,S)
&\le
I(X;Y|U,S),
\label{eq:c2_outer_first}\\
I(U,Z;\overline V|Y,S)
&\le
I(X;Y,Z|U,S),
\label{eq:c2_outer_second}
\end{align}
where
$P_{USXYZ\overline V}
=
P_{US}P_{X|U}P_{YZ|SX}P_{\overline V|U,X,Z}$,
$P_{X|U}\in\mathcal Q_B\cap\mathcal Q_{D_s}$, and
$P_{\hat U|U}\in\mathcal P_{D_u}$ for a reconstruction $\hat U$
based on $(\overline V,Y,S)$.
\end{theorem}

\begin{proof}
We briefly outline the converse.
For each $i$, define
$\overline V_i=(U^{i-1},Y^{i-1},Y_{i+1}^n,S^{-i})$.
The model assumptions imply
$\overline V_i-(U_i,X_i,Z_i)-(Y_i,S_i)$.
Entropy chain rules and time sharing yield
\eqref{eq:c2_outer_first}--\eqref{eq:c2_outer_second},
with a finite cardinality bound permitting passage to the
asymptotic constraints. The complete derivation is given in
\appref{F}{app:c2_outer}.
\end{proof}

\section{Numerical Example}
\label{sec:numerical_examples}
In this section, a Case~2 example shows that the scheme in
Theorem~\ref{thm:c2_wz_cc} strictly outperforms the SSCC benchmark.
Theorem~\ref{thm:c2_inner} further improves upon
Theorem~\ref{thm:c2_wz_cc} and matches the converse in
Theorem~\ref{thm:c2_outer} at the specified budgets.

Let $U\sim{\rm Bernoulli}(0.1)$, $N\sim{\rm Bernoulli}(1/4)$,
and $K\sim{\rm Bernoulli}(1/2)$ be independent, and set
$T=U\oplus N$, $S=(T,K)$, $Y=X$, with $Z=(T,K)$ when $X=T$
and $Z=(T,*)$ otherwise; Hamming distortions on $U$ and on
the state pair $S$ are used, with $b(x)=x$.
The optimal sensing cost is
$c(u,x)=\frac18+\frac14\mathbf1\{x\ne u\}$.

Fix $(D_s,B)=(0.155,0.05)$, so admissible inputs satisfy
$\Pr\{X\ne U\}\le0.12$ and $\Pr\{X=1\}\le0.05$.
Let $H_b$ denote the binary entropy function. For
$0\le D_u\le0.05$, the Wyner--Ziv rate--distortion function is
\begin{equation}
r_{\rm WZ}(D_u)
:=
H_b(0.1)-H_b(0.3)
+H_b\!\left(\frac14+\frac{D_u}{2}\right)-H_b(D_u).
\label{eq:c2_example_wz}
\end{equation}
At the specified budgets, the SSCC benchmark has
$C_{\rm sep}(D_s,B)=H_b(1/40)\approx0.168661$.

With source--state correlation, the SSCC benchmark may not be optimal. We evaluate the
SSCC benchmark and
Theorems~\ref{thm:c2_wz_cc}--\ref{thm:c2_outer} through the
following conditions:
\begin{subequations}
\label{eq:c2_example_conditions}
\begin{align}
r_{\rm WZ}(D_u)
&\le C_{\rm sep}(D_s,B),
\label{eq:c2_example_sscc_condition}\\*
r_{\rm WZ}(D_u)
&\le I(A;Y,S)-I(A;U),
\label{eq:c2_example_condition_t2}\\*
I(U,X,Z;\widetilde V|A,Y,S)
&\le I(A;Y,S)-I(A;U),
\label{eq:c2_example_condition_t3}\\*
H(U|S)-H_b(D_u)
&\le\max_{P_{X|U}\in\mathcal Q_B\cap\mathcal Q_{D_s}}H(X|T).
\label{eq:c2_example_condition_t4}
\end{align}
\end{subequations}
To evaluate Theorem~\ref{thm:c2_wz_cc}, we let
$D_u=0.03$. Let $A=X$ and choose
$P(X=1|U=0)=7/180$ and $P(X=1|U=1)=3/20$.
This choice maximizes the right-hand side of
\eqref{eq:c2_example_condition_t2} under the stated constraints,
with $I(A;Y,S)-I(A;U)=H(X|U)$.
By Theorem~\ref{thm:c2_wz_cc}, an achievable condition is
\begin{equation}
\begin{gathered}
r_{\rm WZ}(D_u)\le H(X|U)\\
=\frac9{10}H_b\!\left(\frac7{180}\right)
+\frac1{10}H_b\!\left(\frac3{20}\right)\approx0.274441.
\end{gathered}
\label{eq:c2_example_rcc}
\end{equation}
Since $r_{\rm WZ}(0.03)\approx0.227511$,
\eqref{eq:c2_example_sscc_condition} is violated, whereas
\eqref{eq:c2_example_condition_t2} holds.
Theorem~\ref{thm:c2_wz_cc} therefore certifies a strict gain
over the specified SSCC benchmark.

To evaluate Theorem~\ref{thm:c2_inner}, we let
$D_u=0.018$, for which $r_{\rm WZ}(0.018)\approx0.282879$.
This exceeds the optimized rate in
\eqref{eq:c2_example_rcc}, so
\eqref{eq:c2_example_condition_t2} does not hold.
Retaining the same $P_{X|U}$ and $A=X$, choose the binary
supplementary description with $\hat U=\widetilde V$
and $P(U\ne\widetilde V|\widetilde V,X,T)=D_u$.
For $0\le D_u\le17/536$, it gives
\begin{equation}
\begin{aligned}
I(U,X,Z;\widetilde V|A,Y,S)
&=H(U|X,T)-H_b(D_u).
\end{aligned}
\label{eq:c2_example_inner}
\end{equation}
At $D_u=0.018$, this description rate is
approximately $0.258688$, which is below $0.274441$.
Thus, \eqref{eq:c2_example_condition_t3} holds.
By Theorem~\ref{thm:c2_inner}, this point is achievable,
establishing a strict improvement over the inner bound in
Theorem~\ref{thm:c2_wz_cc}.
The supplementary description exploits decoded $A=Y=X$ and
feedback-revealed $T$ as side information, yielding the gain.

To evaluate Theorem~\ref{thm:c2_outer}, the
maximization in \eqref{eq:c2_example_condition_t4} over all
feasible $P_{X|U}$ gives
\[
\begin{gathered}
\max_{P_{X|U}\in\mathcal Q_B\cap\mathcal Q_{D_s}}H(X|T)\\
=\frac7{10}H_b\!\left(\frac3{70}\right)
+\frac3{10}H_b\!\left(\frac1{15}\right)\approx0.284677.
\end{gathered}
\]
The same $P_{X|U}$ attains this maximum.
By Theorem~\ref{thm:c2_outer},
\eqref{eq:c2_example_condition_t4} is necessary for $D_u\le1/2$.
Under this input law,
\[
H(U|X,T)-H(X|U)=H(U|S)-H(X|T).
\]
\Needspace{6\baselineskip}
Hence the sufficient condition
\eqref{eq:c2_example_condition_t3} and the optimized necessary
condition \eqref{eq:c2_example_condition_t4} reduce to the same
entropy inequality:
\begin{equation}
\begin{aligned}
&H_b\!\left(\frac1{10}\right)+H_b\!\left(\frac14\right)
-H_b\!\left(\frac3{10}\right)\\
&\quad-H_b(D_u)\le
\frac7{10}H_b\!\left(\frac3{70}\right)
+\frac3{10}H_b\!\left(\frac1{15}\right).
\end{aligned}
\label{eq:c2_example_converse_condition}
\end{equation}
The converse therefore requires $D_u\ge D_u^\star$, where
\begin{equation}
\begin{gathered}
D_u^\star=H_b^{-1}\!\Bigg(
H_b\!\left(\frac1{10}\right)+H_b\!\left(\frac14\right)\\
-H_b\!\left(\frac3{10}\right)
-\frac7{10}H_b\!\left(\frac3{70}\right)
-\frac3{10}H_b\!\left(\frac1{15}\right)\Bigg).
\end{gathered}
\label{eq:c2_example_optimal_distortion}
\end{equation}
Here $H_b^{-1}$ is the inverse on $[0,1/2]$.
Since $D_u^\star<17/536$, the supplementary description
applies at the boundary, where
\[
H(U|X,T)-H_b(D_u^\star)=H(X|U).
\]
By Theorem~\ref{thm:c2_inner}, $D_u^\star$ is achievable
and attains the converse bound in Theorem~\ref{thm:c2_outer},
establishing optimality at
$(D_u,D_s,B)=(D_u^\star,0.155,0.05)$.
See \appref{G}{app:c2_examples} for the construction
and maximizations.

\section{Conclusion}
\label{sec:conclusion}
For independent sources and states, we establish the optimality of SSCC under the sensing-distortion and input-cost constraints. For correlated source--state pairs, we derive single-letter inner and outer bounds for joint coding. Evaluation of these bounds demonstrates a strict gain over the SSCC benchmark and shows that the sufficient condition in Theorem~\ref{thm:c2_wz_cc} is not necessary. At the specified budgets, the feedback-assisted inner bound matches the converse, while a general matching characterization remains open.

\Needspace{10\baselineskip}
\appendices
{
\section{Optimal State Estimation}\label{app:state_estimation}
}

\begin{proof}[Proof of Lemma~\ref{Lemma1}]
{Fix a deterministic code with strictly causal
feedback, $X_i=x_i(U^k,Z^{i-1})$, under either source--state
model of Section~\ref{sec:system_model}. All posterior identities
below are asserted on positive-probability conditioning events;
their values elsewhere do not affect the expected distortion.

In both cases, the states are conditionally
independent given $U^k$, so
\[
P_{S^n|U^k}(s^n|u^k)
=\prod_{j=1}^nP_{S_j|U^k}(s_j|u^k).
\]
Given $U^k$, past feedback is generated from past
states through the memoryless channel and the causal encoding
functions. Hence, given $U^k$, $S_i$ is independent of $Z^{i-1}$
and consequently of $X_i$. Also, the channel law gives
$P_{Z_i|S_i,U^k,X_i}=P_{Z|SX}$.
Bayes' rule therefore yields \eqref{eq:prob}. In Case~1, independence of $U^k$ and $S^n$ alone
would not justify removing $U^k$ from a posterior
conditioned on the input; the encoding and channel
constraints above are essential.

To show that the full observation sequences
provide no further information about $S_i$ once
$(U^k,X_i,Z_i)$ is given, marginalize $Y^n$ from the
code-induced joint law to obtain
\[
\begin{aligned}
&P_{U^kS^nX^nZ^n}(u^k,s^n,x^n,z^n)\\
&=P_{U^k}(u^k)\prod_{j=1}^n\Bigl[
P_{S_j|U^k}(s_j|u^k)\\
&\qquad\times\mathbf1\{x_j=x_j(u^k,z^{j-1})\}
P_{Z|SX}(z_j|s_j,x_j)\Bigr].
\end{aligned}
\]
For fixed $(u^k,x^n,z^n)$ in its support, the
encoding indicators are constant with respect to $s^n$ and
cancel when forming the posterior. Its normalization separates
into one sum for each state coordinate, giving
\[
\begin{aligned}
&P_{S^n|U^kX^nZ^n}(s^n|u^k,x^n,z^n)\\
&=\prod_{j=1}^n
\frac{P_{S_j|U^k}(s_j|u^k)P_{Z|SX}(z_j|s_j,x_j)}
{\sum_{\tilde s}P_{S_j|U^k}(\tilde s|u^k)
P_{Z|SX}(z_j|\tilde s,x_j)}.
\end{aligned}
\]
Thus,
\begin{align}
&P_{S^n|U^kX^nZ^n}(s^n|u^k,x^n,z^n)\nonumber\\
&\qquad=\prod_{i=1}^nP_{S_i|U^kX_iZ_i}(s_i|u^k,x_i,z_i),
\label{eq:state_posterior_factorization}
\end{align}
with each factor given by \eqref{eq:prob}.
Marginalizing all states except $S_i$ yields
\[
P_{S_i|U^kX^nZ^n}=P_{S_i|U^kX_iZ_i}.
\]
For any admissible estimator, $\hat S^n$ is a function of $(U^k,X^n,Z^n)$.
Conditioning on these observations fixes
$\hat S_i$. The preceding posterior identity and pointwise
minimization of the conditional expected distortion give
\begin{align}
&\mathbb E[d_s(S_i,\hat S_i)]\nonumber\\
&=\mathbb E\!\left[\sum_{s\in\mathcal S}
P_{S_i|U^kX_iZ_i}(s|U^k,X_i,Z_i)d_s(s,\hat S_i)\right]\nonumber\\
&\geq\mathbb E\!\left[\min_{s'\in\hat{\mathcal S}}\sum_{s\in\mathcal S}
P_{S_i|U^kX_iZ_i}(s|U^k,X_i,Z_i)d_s(s,s')\right]\nonumber\\
&=\mathbb E[d_s(S_i,\hat s_i^*(U^k,X_i,Z_i))].
\label{eq:state_estimator_optimality}
\end{align}
Summing over $i$ and dividing by $n$ proves the claim.
In Case~1, substituting
$P_{S_i|U^k}(s|u^k)=P_S(s)$ into \eqref{eq:prob}
makes the posterior depend only on $(x_i,z_i)$.
The same fixed tie-breaking rule therefore gives
$\hat s_i^*(u^k,x_i,z_i)=\hat s^*(x_i,z_i)$.
Although the lower bound was proved after also providing
$U^k$ to the estimator, this minimizing rule uses only
$(X^n,Z^n)$ and is admissible under the original information
constraint. In Case~2, the i.i.d. pair model instead gives
$P_{S_i|U^n}(s|u^n)=P_{S|U}(s|u_i)$, so the minimizing
rule is $\hat s^*(u_i,x_i,z_i)$.}
\end{proof}

{
\section{Proof of Theorem~\ref{main_case1}}\label{app:main_case1}
}
\begin{proof}
First consider rational $\gamma>0$ and blocklengths
satisfying $k/n=\gamma$.

1) Converse:
{For each $i\in[1:k]$, define
\begin{equation}
D_{u,i}:=\mathbb{E}[d_u(U_i,\hat U_i)].
\label{eq:dui_case1}
\end{equation}}
{Since $R(D_u)$ is convex and $\Delta_u^{(k)}=\frac{1}{k}\sum_{i=1}^kD_{u,i}$, we have
    \begin{subequations}\label{eq:main_derivation_case1}
    \begin{align}
R\bigl(\Delta_u^{(k)}\bigr)
&=R\left(\frac{1}{k}\sum_{i=1}^kD_{u,i}\right)\nonumber\\
&\leq \frac{1}{k}\sum_{i=1}^kR(D_{u,i}) \label{eq:step4_case1}\\
&\leq \frac{1}{k}\sum_{i=1}^k I(U_i;\hat U_i) \label{eq:step5_case1}\\
&\leq \frac{1}{k}I(U^k;\hat U^k)\nonumber\\
&\leq \frac{1}{k}I(U^k;\hat U^k,S^n)\nonumber\\
&=\frac{1}{k}I(U^k;\hat U^k|S^n)\nonumber\\
&\leq \frac{1}{k}I(U^k;Y^n|S^n). \label{eq:source_channel_bridge_case1}
    \end{align}
    \end{subequations}}
{Here, \eqref{eq:step5_case1} follows from the definition of $R(D_u)$: the joint distribution of $(U_i,\hat U_i)$ induces a test channel $P_{\hat U_i|U_i}$ with distortion $D_{u,i}$, and $R(D_{u,i})$ is the minimum of $I(U;\hat U)$ over all test channels with distortion at most $D_{u,i}$.  Moreover,
\begin{align}
I(U^k;\hat U^k)
&=H(U^k)-H(U^k|\hat U^k)\nonumber\\
&\geq\sum_{i=1}^k\bigl[H(U_i)-H(U_i|\hat U_i)\bigr],
\label{eq:block_source_mi_case1}
\end{align}
which proves the fourth line of \eqref{eq:main_derivation_case1}.  The equality preceding \eqref{eq:source_channel_bridge_case1} follows from the independence of $U^k$ and $S^n$, and its last inequality follows because $\hat U^k$ is a function of $(Y^n,S^n)$.}

{Therefore, using $\gamma=k/n$, we have
    \begin{subequations}\label{eq:channel_mi_bound}
    \begin{align}
\gamma R\bigl(\Delta_u^{(k)}\bigr)
&\leq \frac{1}{n}I(U^k;Y^n|S^n) \label{eq:channel}\\
&=\frac{1}{n}\sum_{i=1}^n I(U^k;Y_i|Y^{i-1},S^n) \label{eq:channel_mi_chain}\\
&\leq\frac{1}{n}\sum_{i=1}^n I(U^k,X_i;Y_i|Y^{i-1},S^n)\nonumber\\
&=\frac{1}{n}\sum_{i=1}^n I(X_i;Y_i|Y^{i-1},S^n) \label{eq:channel_mi_memoryless}\\
&\leq\frac{1}{n}\sum_{i=1}^n I(X_i;Y_i|S_i). \label{eq:channel_mi_single_letter}
    \end{align}
    \end{subequations}}
{The equality in \eqref{eq:channel_mi_memoryless} follows from the memoryless channel law: conditioned on $(X_i,S_i)$, the current output $Y_i$ is independent of $(U^k,Y^{i-1},S^{i-1},S_{i+1}^n)$.  The last inequality follows from
\begin{align}
H(Y_i|Y^{i-1},S^n)&\leq H(Y_i|S_i),\nonumber\\
H(Y_i|X_i,Y^{i-1},S^n)&=H(Y_i|X_i,S_i).
\label{eq:feedback_single_letter_case1}
\end{align}
Thus, this step does not condition on future channel inputs and remains valid when $X_i$ depends on the strictly causal feedback $Z^{i-1}$.}

{Let $P_{X_i}$ denote the marginal input distribution induced by the code, and define
\begin{align}
D_{s,i}&:=\mathbb{E}[d_s(S_i,\hat s^*(X_i,Z_i))]
=\sum_{x\in\mathcal X}P_{X_i}(x)c(x),\nonumber\\
B_i&:=\mathbb{E}[b(X_i)]
=\sum_{x\in\mathcal X}P_{X_i}(x)b(x).
\label{eq:per_letter_constraints_case1}
\end{align}
Since $S_i$ is i.i.d. and is independent of $(U^k,Z^{i-1})$ in Case~1, $X_i$ and $S_i$ are independent.  Hence, the joint distribution associated with $P_{X_i}$ has the form required in \eqref{eq:channelcapacity_case1}.}
{It follows that
    \begin{align}
\gamma R\bigl(\Delta_u^{(k)}\bigr)
&\leq \frac{1}{n}\sum_{i=1}^n C_{\inf}(D_{s,i},B_i) \label{eq:channel_mi_cinf}\\
&\leq C_{\inf}\left(\frac{1}{n}\sum_{i=1}^nD_{s,i},
\frac{1}{n}\sum_{i=1}^nB_i\right) \label{eq:channel_mi_jensen}\\
&\leq C_{\inf}\bigl(\Delta_s^{(n)},
\tfrac{1}{n}\sum_{i=1}^n\mathbb{E}[b(X_i)]\bigr). \label{eq:channel_mi_final}
    \end{align}}
{Here, \eqref{eq:channel_mi_cinf} follows from the definition of $C_{\inf}(D_s,B)$, and \eqref{eq:channel_mi_jensen} follows from the joint concavity of $C_{\inf}(D_s,B)$ \cite[Lemma~2]{ahmadipour2022information}.  By Lemma~\ref{Lemma1}, the symbol-by-symbol estimator $\hat s^*(X_i,Z_i)$ minimizes the sensing distortion, and hence $n^{-1}\sum_{i=1}^nD_{s,i}\leq\Delta_s^{(n)}$; the last inequality then follows from the monotonicity of $C_{\inf}(D_s,B)$.  Finally, applying the three limits in \eqref{eq:achievable_conditions} and the continuity and monotonicity of the finite-alphabet functions $R(D_u)$ and $C_{\inf}(D_s,B)$ yields
\begin{equation}
\gamma R(D_u)\leq C_{\inf}(D_s,B).
\label{eq:case1_converse_result}
\end{equation}}

{
2) \textit{Achievability part:} We use separate lossy source coding and channel coding. Fix $P_X\in\mathcal P_{D_s}^{(1)}\cap\mathcal P_B$ and local rates $R_s,R$ such that
\begin{equation}
R(D_u)<R_s,\qquad \gamma R_s<R<I(X;Y|S).
\label{eq:case1_rate_matching}
\end{equation}
Such a choice is possible whenever $\gamma R(D_u)<C_{\inf}(D_s,B)$. Choose integer blocklengths with $k/n\to\gamma$, taking $k/n=\gamma$ when $\gamma$ is rational. Define $P_{SXY}:=P_SP_XP_{Y|SX}$. Independently generate $\lfloor2^{nR}\rfloor$ channel codewords $x^n(w)$ according to $\prod_{i=1}^nP_X(x_i)$; the resulting codebook $\mathcal C$ is known to both terminals.

a) \textit{Source encoding:} Fix $\epsilon>0$. By the classical rate-distortion coding theorem \cite[Ch.~10]{cover1999elements}, choose a deterministic source code with at most $\lceil2^{kR_s}\rceil$ indices and average distortion at most $D_u+\epsilon$, for all sufficiently large $k$. Its encoder produces $W=f(U^k)$, and its reconstruction for index $w$ is $\hat u^k(w)$. The strict rate inequality in \eqref{eq:case1_rate_matching} allows these indices to be embedded in $[1:\lfloor2^{nR}\rfloor]$ for all sufficiently large $n$.

b) \textit{Channel coding:} To send message $w$, the encoder transmits $x^n(w)$. Given $(y^n,s^n)$, the decoder seeks a unique index $\hat w$ satisfying
\begin{equation}
(s^n,x^n(\hat w),y^n)\in T_\delta^{(n)}(P_{SXY}),
\label{eq:case1_channel_decoder}
\end{equation}
where $\delta>0$ is sufficiently small for the strict packing inequality in \eqref{eq:case1_rate_matching}. It declares an error if no unique index exists.

c) \textit{Source decoding:} Upon obtaining a valid source index $\hat w$, the source decoder outputs $\hat u^k(\hat w)$; otherwise, it outputs a fixed sequence in $\hat{\mathcal U}^k$.

d) \textit{Estimation:} The transmitter computes
\begin{equation}
\hat S_i=\hat s^*(X_i,Z_i),\qquad i\in[1:n].
\label{eq:case1_ach_sensing_estimator}
\end{equation}
This construction uses feedback only for sensing.

e) \textit{Analysis and codebook selection:} The standard packing lemma
\cite[Lemma~3.1]{el2011network}
gives a vanishing ensemble-average channel error probability. Although $W$ need not be uniform, the random channel codebook is independent of the fixed source code, and the ensemble error probability is the same for every transmitted index.

To select one codebook satisfying the average sensing-distortion
and input-cost constraints, define the joint bad event
\begin{equation}
\begin{aligned}
\mathcal E_n:= {}&\{\hat W\ne W\}\\
&\cup\bigcup_{\ell\in\{b,c\}}
\left\{\frac1n\sum_{i=1}^n\ell(X_i)
>\mathbb E[\ell(X)]+\epsilon\right\},
\end{aligned}
\label{eq:case1_joint_bad_event}
\end{equation}
where a channel decoding error is included in $\{\hat W\ne W\}$. For each transmitted index, its codeword is i.i.d. under the ensemble. Thus, the law of large numbers and the packing bound imply $\mathbb E_{\mathcal C}[\Pr(\mathcal E_n\mid\mathcal C)]\to0$, also under the actual distribution of $W$. Averaging therefore yields deterministic codebooks $\mathcal C_n$ for which $q_n:=\Pr(\mathcal E_n\mid\mathcal C_n)\to0$. All probabilities and expectations below refer to these fixed codes.

Let $\ell_{\max}:=\max_x\ell(x)$ for $\ell\in\{b,c\}$ and $d_{\max}:=\max_{u,\hat u}d_u(u,\hat u)$. These constants are finite. Since the source and state are independent and the channel encoder ignores feedback,
\begin{align}
\Delta_s^{(n)}
&=\frac1n\sum_{i=1}^n\mathbb E[c(X_i)]
\le D_s+\epsilon+c_{\max}q_n,
\label{eq:case1_sensing_distortion_ach}\\
\frac1n\sum_{i=1}^n\mathbb E[b(X_i)]
&\le B+\epsilon+b_{\max}q_n,
\label{eq:case1_input_cost_ach}\\
\Delta_u^{(k)}
&\le D_u+\epsilon+d_{\max}q_n.
\label{eq:case1_source_distortion_ach}
\end{align}
The last bound compares the received reconstruction with that of the fixed source code and uses $\Pr(\hat W\ne W)\le q_n$. Letting $n\to\infty$ and then $\epsilon\downarrow0$, with a diagonal choice of blocklengths and codes, establishes all three limsup constraints and proves achievability.
}

For arbitrary blocklengths, the converse above holds with
$\gamma$ replaced by $k/n$, and the necessary condition
follows as $k/n\to\gamma$.
For a general real $\gamma$ satisfying the strict inequality,
rational approximation and a diagonal choice of codes
establish achievability.

To justify the operational separation statement, it remains
to include the boundary $\gamma R(D_u)=C_{\inf}(D_s,B)$.
If this common value is positive, mix an optimal
rate-distortion test channel with a constant reconstruction.
Convexity of $R$ gives strictly smaller source rates at
distortions tending to $D_u$. The strict achievability result
then applies with the same $D_s,B,\gamma$; a diagonal choice
of separated codes attains the boundary tuple.
If the common value is zero, $R(D_u)=0$ and the finite-alphabet
minimum is attained by a reconstruction independent of $U$.
Some constant reconstruction therefore has distortion at most
$D_u$. A deterministic input sequence with empirical
distribution tending to any feasible $P_X$, together with
$\hat s^*(X_i,Z_i)$, meets the sensing and input-cost
constraints without sending a source description.
Thus separation attains every tuple permitted by the converse.

\end{proof}

{
\allowdisplaybreaks[0]
\section{Adaptation of the SSCC Benchmark Channel Capacity}
\label{app:c2_sep}

Consider a message $W$ uniformly distributed over its message set and an encoder $X_i=x_i(W,Z^{i-1})$. Under the Case~2 SSCC benchmark, the current source--state block is generated independently of $W$, so that
\[
P_{WU^nS^n}=P_W P_{U^nS^n}.
\]
Operationally, the message sent in each block
describes the preceding source block and is independent of the
current source--state block. The receiver uses the state sequence
associated with the described source block for its reconstruction.
The block construction and its vanishing boundary overhead are
detailed in \appref{D}{app:c2_cc}.
Since the source--state pairs are memoryless,
$(U_i,S_i)$ is independent of $(W,Z^{i-1})$.
Thus, $X_i$ is independent of $(U_i,S_i)$.

Values of $c(u,x)$ at $P_U(u)=0$ have no effect on any constraint.

The posterior factorization in Lemma~\ref{Lemma1}
continues to hold under the SSCC benchmark above. Even if the state estimator is
additionally given $W$,
\[
P(S_i|W,U^n,X^n,Z^n)=P(S_i|U_i,X_i,Z_i).
\]
The minimum average sensing distortion is therefore
$n^{-1}\sum_i\mathbb E[c(U_i,X_i)]$.

For achievability, fix a feasible $P_X$ and use the separated
channel code described above. Treating $(Y,S)$
as the channel output, every rate
\[
R_{\rm ch}<I(X;Y,S)=I(X;Y|S)
\]
is achievable. Constant-composition input codewords
can approximate $P_X$, and the transmitter uses
$\hat s^*(U_i,X_i,Z_i)$.
For every codeword, its composition determines the expected
input cost and sensing cost under the independent current
source--state block. Thus both converge to the prescribed
expectations. Expurgation gives vanishing maximal error without
changing these costs or the limiting rate. If the capacity is
zero, a single deterministic composition sequence realizes
rate zero and the feasible costs.

For the converse, Fano's inequality
\cite[Sec.~7.9]{cover1999elements} gives
\begin{equation}
nR_{\rm ch}
\le I(W;Y^n|S^n)+n\varepsilon_n,
\qquad \varepsilon_n\to0.
\end{equation}
Using the entropy chain rule and the fact that conditioning
cannot increase entropy,
\begin{align}
&I(W;Y^n|S^n)\nonumber\\
&=\sum_{i=1}^n\Big[
H(Y_i|Y^{i-1},S^n)
-H(Y_i|W,Y^{i-1},S^n)
\Big]\nonumber\\
&\le\sum_{i=1}^n\Big[
H(Y_i|S_i)
-H(Y_i|W,Y^{i-1},S^n,Z^{i-1})
\Big]\nonumber\\
&=\sum_{i=1}^n I(X_i;Y_i|S_i).
\label{eq:c2_sep_converse}
\end{align}
For the last equality, $X_i$ is a deterministic function of
$(W,Z^{i-1})$. By the memoryless channel law,
$Y_i-(X_i,S_i)-(W,Y^{i-1},S^{-i},Z^{i-1})$ forms a Markov
chain. Consequently,
$H(Y_i|W,Y^{i-1},S^n,Z^{i-1})=H(Y_i|X_i,S_i)$,
which yields the last equality in
\eqref{eq:c2_sep_converse}.

Let
\[
\bar P_X^{(n)}:=\frac1n\sum_{i=1}^n P_{X_i}.
\]
By concavity of mutual information in the input distribution,
\[
R_{\rm ch}
\le I_{\bar P_X^{(n)}}(X;Y|S)+\varepsilon_n,
\]
where the mutual information is evaluated under
$P_{US}\bar P_X^{(n)}P_{YZ|SX}$.
Moreover,
\begin{align*}
\sum_{u,x}P_U(u)\bar P_X^{(n)}(x)c(u,x)
&\le \Delta_s^{(n)},\\
\sum_x\bar P_X^{(n)}(x)b(x)
&=\frac1n\sum_{i=1}^n\mathbb E[b(X_i)].
\end{align*}
The first inequality follows from the optimality of the
symbol-by-symbol state estimator.

For an asymptotically feasible code sequence, compactness
of the probability simplex yields a convergent subsequence
of $\bar P_X^{(n)}$, with limit $P_X^\star$.
The limsup constraints and continuity of the costs imply
that $P_X^\star$ is feasible at $(D_s,B)$.
Continuity of mutual information and $\varepsilon_n\to0$
therefore give
\[
R_{\rm ch}
\le I_{P_X^\star}(X;Y|S)
\le C_{\rm sep}(D_s,B).
\]

\section{Proof of Theorem~\ref{thm:c2_wz_cc}}
\label{app:c2_cc}

\subsection{Cover--Chiang Channel Coding}

Under the joint distribution in
Theorem~\ref{thm:c2_wz_cc}, regard $U$ as the encoder-side
state and $S$ as the correlated decoder-side state in the
Cover--Chiang model \cite{cover2002duality}. The induced
memoryless channel is
$P_{S,Y|U,X}=P_{S|U}P_{Y|S,X}$.
Hence the Cover--Chiang rate specializes to
$I(A;Y,S)-I(A;U)$.

Fix a distribution with the factorization in
Theorem~\ref{thm:c2_wz_cc}
and the sensing-distortion and input-cost constraints.
Generate $2^{n(R_{\rm ch}+R_L)}$ independent
$A^n$ codewords, partitioned into $2^{nR_{\rm ch}}$
groups of size $2^{nR_L}$.
The independent message identifies a group.
Within the indicated group, the encoder selects the first
codeword jointly typical with $U^n$ in a fixed ordering,
without inspecting any other group. By the standard covering lemma
\cite[Lemma~3.3]{el2011network}, the covering step succeeds if
\[
R_L>I(A;U).
\]
Here $R_{\rm ch}$ and $R_L$ are local message and covering
rates. For every possible source sequence, message, and
selected codeword, generate an input sequence at code design
according to $\prod_iP_{X|U,A}$. After all generated input
sequences are fixed, the encoder is deterministic and uses
no feedback. Use a fixed fallback index if covering fails.

The receiver searches for a unique auxiliary codeword
jointly typical with $(Y^n,S^n)$ and recovers its group index.
The dependence introduced by auxiliary-codeword selection
is handled as follows.

Suppose first that $I(A;U)>0$.
For a typical $u^n$, let
\[
p_n(u^n):=
P_A^n\bigl\{a^n:(u^n,a^n)\text{ is jointly typical}\bigr\}.
\]
The joint typicality bound gives
\[
p_n(u^n)
\le 2^{-n[I(A;U)-\delta(\epsilon)]},
\]
where $\delta(\epsilon)\to0$ as the typicality tolerance
$\epsilon\downarrow0$.
Fix $\epsilon$ sufficiently small.

Let $J$ denote the selected position within the transmitted
group. Condition on $U^n=u^n$, successful covering,
$J=j$, and the selected codeword.
Candidates after position $j$ retain their product law.
Each candidate before position $j$ has the product law
conditioned on failing the source-typicality test.
Its probability mass is therefore at most
$(1-p_n(u^n))^{-1}$ times its original mass, which is
at most $2$ for all sufficiently large $n$.
Candidates in other groups retain their product law.

The encoder uses only the corresponding input sequence
generated at code design.
Thus, conditional on the source sequence and selected
codeword, the same bounds apply when conditioning further
on the channel observations.
The joint typicality bound and a union bound over all
incorrect candidates consequently give
\[
\begin{aligned}
P_{\rm e}^{(n)}
\le {}&o(1)+2\,
2^{n(R_{\rm ch}+R_L)}\\
&\quad{}\times
2^{-n[I(A;Y,S)-\delta'(\epsilon)]},
\end{aligned}
\]
where $P_{\rm e}^{(n)}$ is averaged over the random code
design and message, $\delta'(\epsilon)\to0$, and the
$o(1)$ term includes covering failure and failure of
joint typicality of the transmitted sequences.

If $I(A;U)=0$, the first codeword in the indicated group
is jointly typical with $U^n$ with probability tending
to one. On this event, all remaining candidates retain
their independent product laws; the complementary event
has vanishing probability. The same packing condition follows.

Hence, using the standard typicality and packing bounds
\cite[Section~2.5 and Lemma~3.1]{el2011network},
the decoding error probability vanishes whenever
\[
R_{\rm ch}+R_L<I(A;Y,S).
\]
Thus every rate
\[
R_{\rm ch}<I(A;Y,S)-I(A;U)
\]
is achievable.
Joint typicality ensures that the input cost and
average $c(U,X)$ approach their prescribed expectations.
The local optimal estimator converts the latter
into the actual sensing distortion.

For each message, consider the union of decoding failure and
failure of joint typicality of $(U^n,S^n,A^n,X^n,Y^n,Z^n)$.
Its probability, averaged over messages and the random code
design (including the generated input sequences), tends to zero. Fix a design
with vanishing average probability of this single bad event.
If that average is $e_n$, discard messages whose bad-event
probability exceeds $\sqrt{e_n}$; at most a fraction
$\sqrt{e_n}$ is discarded. Generate at a slightly larger
message rate, still strictly below the achievable expression,
to retain the required number of messages. On the complementary
event, both empirical costs differ from their target expectations
by a vanishing tolerance. Boundedness controls the bad event.
Consequently, the retained code satisfies both average constraints
uniformly over its messages and can transmit any nonuniform
index independent of the current source--state block.

For the maximization of $I(A;Y,S)-I(A;U)$
over feasible conditional pmfs $P_{A,X|U}$, it suffices to consider
$|\mathcal A|\le|\mathcal U||\mathcal X|+1$:
preserving $P_{UX}$ and
$H(U|A)-H(Y,S|A)$ preserves the objective, expected sensing distortion, and
input cost, and the support lemma
\cite[Appendix~C]{el2011network}
gives the stated bound.

\subsection{JSCC Construction}

Suppose \eqref{eq:c2_wz_cc_condition} holds strictly for a
feasible conditional pmf $P_{A,X|U}$. Fix a conditional pmf $P_{V|U}$ and
reconstruction $\phi(V,S)$ attaining $R_{U|S}(D_u)$, and use
the given $P_{A,X|U}$ under \eqref{eq:c2_wz_cc_law}.
Choose a rate $R$ such that
\[
I(U;V|S)<R<I(A;Y,S)-I(A;U).
\]
Partition the source and channel sequences into $L$ blocks
of length $m$. By Wyner--Ziv coding \cite{wyner1976rate},
source block $j$ can be described by an index of rate $R$
and reconstructed using that index and its state block,
with expected distortion at most
$\mathbb E[d_u(U,\phi(V,S))]+o(1)$.
Transmit this index in channel block $j+1$ using the
preceding Cover--Chiang construction. The index depends
only on source block $j$ and is independent of the
source--state pair in block $j+1$. The encoder uses the
current source block noncausally to select its channel
input. The uniform message guarantees above therefore
apply to the possibly nonuniform source index.

After recovering the index from $(Y^m,S^m)$ in block $j+1$,
the receiver reconstructs source block $j$ using the
corresponding decoder side-information sequence $S^m$. Bounded source distortion makes the additional
distortion due to index-decoding errors vanish. The same
channel code controls sensing distortion and input cost
uniformly over incoming indices, with state estimates
$\hat s^*(U,X,Z)$. Thus all three expected constraints
hold simultaneously for fixed source and channel codes.

Initialize the first block with a dummy index and reconstruct
the final source block arbitrarily. The latter contributes
at most $\max_{u,\hat u}d_u(u,\hat u)/L$ to average source distortion.
For fixed $L$, let $m\to\infty$, and then let $L\to\infty$.
A diagonal choice gives the required code sequence.
Each physical block has $m$ source symbols and $m$ channel
uses, so $\gamma=1$ throughout. This proves the strict
achievability statement of Theorem~\ref{thm:c2_wz_cc}.
A further diagonal selection over convergent achievable
distortion--cost tuples proves achievability of its closure:
for tuples converging to $(D_u,D_s,B)$, select one code for
each tuple with all three expected constraints within a
vanishing tolerance and with strictly increasing blocklengths.
The resulting sequence satisfies Definition~\ref{def:achievable}.

\subsection{Boundary Cases}
\label{app:c2_cc_boundary}

It remains to treat equality in \eqref{eq:c2_wz_cc_condition}
for the given $P_{A,X|U}$. The finite-alphabet Wyner--Ziv
minimum is attained. Fix a minimizing $P_{V|U}$ and the
given $P_{A,X|U}$, and let $r$ be their common information
value. If $r$ is positive,
for $0<\epsilon<1$, let $Q\sim\mathrm{Bernoulli}(\epsilon)$
be independent of $(U,S,V)$. Use the source description
$V_\epsilon=(Q,V)$ when $Q=0$ and $(Q,v_0)$ when $Q=1$,
with a fixed reconstruction on the latter branch. Then
\[
I(U;V_\epsilon|S)=(1-\epsilon)r<r,
\]
while its expected distortion converges to that of the original
reconstruction. This description satisfies $V_\epsilon-U-S$;
$Q$ is part of its auxiliary alphabet and requires no shared
randomness. The conditional pmf $P_{A,X|U}$, expected sensing distortion,
and input cost are unchanged.
Strict achievability followed by the diagonal argument above
therefore establishes the target tuple.

If $r=0$, an optimal Wyner--Ziv description satisfies
$I(U;V|S)=0$. Consequently, the reconstruction
\[
\psi(s)\in\arg\min_{\hat u}
\sum_uP_{U|S}(u|s)d_u(u,\hat u)
\]
has expected distortion no greater than that of $\phi(V,S)$.
Thus the receiver can reconstruct from $S^n$ alone.
Choose any feasible $P_{X|U}$. Every such conditional pmf is a convex
combination of deterministic maps from $\mathcal U$ to
$\mathcal X$. Fix a deterministic schedule of these maps
whose empirical proportions converge to the combination
weights, and use $X_i=f_i(U_i)$. The average input cost and
average $c(U,X)$ converge to those of the chosen $P_{X|U}$.
The estimator $\hat s^*(U_i,X_i,Z_i)$ realizes this sensing
cost. This deterministic code satisfies all three limsup
constraints without transmitting a description, including
when no positive communication rate is possible.

\subsection{Comparison with the SSCC Benchmark}

For every input pmf $P_X$ feasible for the SSCC benchmark, set
$A=X$ independent of $U$. Under the factorization in
Theorem~\ref{thm:c2_wz_cc}, $X$
is then also independent of $S$, so
\[
I(A;Y,S)-I(A;U)=I(X;Y|S).
\]
The input cost and expected sensing distortion are unchanged,
which proves \eqref{eq:c2_cc_vs_sep}. Whenever
$R_{U|S}(D_u)<C_{\rm sep}(D_s,B)$, an input pmf feasible for the SSCC benchmark satisfies $R_{U|S}(D_u)<I(X;Y|S)$.
The choice $A=X$ therefore satisfies
\eqref{eq:c2_wz_cc_condition}, establishing the region
inclusion claimed in Theorem~\ref{thm:c2_wz_cc}.
The same interblock argument, using the source-independent
channel encoder of \appref{C}{app:c2_sep}, realizes the SSCC benchmark's achievable region under the strict
rate condition.

\section{Proof of the Feedback JSCC Inner Bound}
\label{app:c2_inner}

First suppose \eqref{eq:c2_inner_condition} holds
strictly for finite-alphabet auxiliaries, a joint distribution, and a
reconstruction satisfying the factorization and distortion--cost
constraints in Theorem~\ref{thm:c2_inner}. There is then a rate $R$ such that
\begin{equation}
I(U,X,Z;\widetilde V|A,Y,S)
<R<I(A;Y,S)-I(A;U).
\label{eq:c2_inner_rate_choice}
\end{equation}

Partition the source and channel sequences into $L$
blocks of length $m$, using independent codebooks
across blocks. Block $j$ carries the bin index of
the supplementary description for block $j-1$.
This index depends only on earlier source blocks and feedback,
and is independent of the source--state pair in block $j$.
The encoder deliberately ignores future source blocks even though
it observes the entire source sequence. The index need not be
independent of the source block that it describes.

\subsection{Current-Block Channel Coding}

Generate $2^{m(R+R_L)}$ codewords $A^m$ and partition
them into $2^{mR}$ groups.
The previous block's bin index identifies a group.
Using the current $U^m$, the encoder selects a
jointly typical codeword in that group and forms
the input according to $P_{X|U,A}$. Use the deterministic
encoding function described in \appref{D}{app:c2_cc}; the
current block's inputs do not depend on its own feedback.

Use the first-typical-codeword selection rule of
\appref{D}{app:c2_cc}.
The decoding analysis there recovers both $A^m$ and
the previous block's bin index from the current
$(Y^m,S^m)$.
Sufficient rate conditions are
\begin{equation}
R_L>I(A;U),\qquad
R+R_L<I(A;Y,S).
\label{eq:c2_inner_channel_rates}
\end{equation}

\subsection{Supplementary Description from Feedback}

For each $A^m$ codeword, independently generate
$2^{m\widetilde R}$ conditional description codewords
according to $\prod_{i=1}^mP_{\widetilde V|A}$,
and assign them randomly to $2^{mR}$ bins.

At the end of the block, the transmitter knows
$(U^m,A^m,X^m,Z^m)$.
It selects a description codeword jointly typical
with these sequences, by a fixed rule independent of the
random binning map, and sends its bin index in the next block.
If no such codeword exists, it uses a fixed fallback index
and declares a covering failure.
Selection of $A^m$ and the input sequences generated at code
design are independent of the supplementary codebooks.
The current-block inputs do not use current-block feedback.
Consequently, on every positive-probability conditioning event,
\[
\begin{aligned}
&P(s^m,y^m\mid u^m,a^m,x^m,z^m)\\
&\qquad=\prod_{i=1}^m P(s_i,y_i\mid u_i,x_i,z_i),
\end{aligned}
\]
where
\[
P(s,y\mid u,x,z)
=
\frac{P_{S|U}(s|u)P_{YZ|SX}(y,z|s,x)}
{\sum_{s',y'}P_{S|U}(s'|u)P_{YZ|SX}(y',z|s',x)}.
\]
The supplementary codebook is generated independently,
and description selection uses only
$(U^m,A^m,X^m,Z^m)$ and that codebook.
Conditioning additionally on the selected
$\widetilde V^m$ therefore leaves this product
conditional law unchanged.

By the conditional form of the covering lemma
\cite[Lemma~3.3]{el2011network},
the covering failure probability tends to zero if
\[
\widetilde R>I(U,X,Z;\widetilde V|A).
\]

On successful covering,
$(U^m,A^m,X^m,Z^m,\widetilde V^m)$ is jointly typical.
The conditional typicality lemma
\cite[Section~2.5]{el2011network}
then yields joint typicality of the full tuple.
After receiving the index, the receiver recovers
the description using the stored $(A^m,Y^m,S^m)$.
Let $\mathcal E_{\rm bin}$ denote the existence of an
incorrect supplementary description in the indicated bin
that is jointly typical with $(A^m,Y^m,S^m)$.
Since description selection does not use the binning map,
averaging over that map and applying a union bound give
\[
\begin{aligned}
\Pr(\mathcal E_{\rm bin})
&\le 2^{-mR}\sum_j
\Pr\!\left\{
\begin{array}{c}
(A^m,\widetilde V^m(j),Y^m,S^m)\\
\text{is jointly typical}
\end{array}
\right\}\\
&\le
2^{m[\widetilde R-R-I(\widetilde V;Y,S|A)
+\delta(\epsilon)]},
\end{aligned}
\]
where $\delta(\epsilon)\to0$.
The last inequality follows from conditional codebook
generation: given $A^m$, each generated description
is independent of $(Y^m,S^m)$.
This is the conditional packing bound
\cite[Lemma~3.1]{el2011network}.
Thus, a sufficient bin-decoding condition is
\[
\widetilde R-R<I(\widetilde V;Y,S|A).
\]

The Markov chain
$\widetilde V-(U,A,X,Z)-(Y,S)$ gives
\begin{align}
&I(U,X,Z;\widetilde V|A)
-I(\widetilde V;Y,S|A)\nonumber\\
&\qquad=I(U,X,Z;\widetilde V|A,Y,S).
\end{align}
Hence \eqref{eq:c2_inner_rate_choice} permits all
covering and packing inequalities to hold
simultaneously.

\subsection{Reconstruction, Constraints, and Limits}

The receiver reconstructs the source symbol by symbol
using $\phi(A,\widetilde V,Y,S)$, and the transmitter
estimates the state using $\hat s^*(U,X,Z)$.
Covering, packing, and conditional typicality ensure
joint typicality of the relevant sequences with
probability tending to one.
On this event, the average source distortion, sensing
distortion, and input cost approach their single-letter
expectations.
The contribution of its complement vanishes
because the costs are bounded.

The ensemble bounds above are identical for every incoming
bin index. Moreover, that index is independent of the random
code design of the current block. Thus they also apply to its
actual, possibly nonuniform, distribution. For fixed $L$, a
union bound over covering, decoding, and joint-typicality
failures in all blocks has expectation tending to zero over
the complete random design. Fix one complete design with this
union-event probability tending to zero. All encoding functions,
codebooks, bin maps, and selection rules are then deterministic.
On the complement of this one event, source distortion in the
first $L-1$ blocks and the sensing distortion and input cost
in all $L$ blocks are
simultaneously within vanishing tolerances of their targets.
This establishes simultaneous constraints for a single code,
rather than selecting a different code for each expectation.

A fixed dummy index initializes the first block.
The source in the last block may be reconstructed
arbitrarily, incurring an average distortion penalty
bounded by a constant divided by $L$.
Every physical block contains the same number of
source symbols and channel uses, so $\gamma=1$
is maintained.

First let $m\to\infty$ for fixed $L$, then let
$L\to\infty$, and fix suitable codebooks and input
mappings. Feedback from a block affects only
subsequent block inputs; the construction therefore
respects strictly causal feedback.
This proves strict achievability.
A diagonal selection of achievable code sequences
establishes achievability of the closure.

\subsection{Boundary Cases}
\label{app:c2_inner_boundary}

Fix a law and reconstruction satisfying the non-strict
condition \eqref{eq:c2_inner_condition} and the three
constraints. First suppose some distribution $P_X$,
independent of $(U,S)$ and not necessarily feasible at the
specified budgets, gives $I(X;Y|S)>0$.
Call the given law branch $0$. In branch $1$, set $A=X$
with this input distribution, take $\widetilde V$ constant,
and use a fixed source reconstruction. For these two laws,
write the information margins as
\[
\begin{aligned}
G_j:={}&I_j(A;Y,S)-I_j(A;U)\\
&-I_j(U,X,Z;\widetilde V|A,Y,S),\qquad j=0,1.
\end{aligned}
\]
Then $G_0\ge0$ and $G_1=I_1(X;Y|S)>0$.

Let $Q\sim\mathrm{Bernoulli}(\epsilon)$ be independent of
$(U,S)$ and select the corresponding conditional pmfs.
Include $Q$ in $A_\epsilon=(Q,A_Q)$, and use the selected
branch's description and reconstruction. This is a finite
auxiliary construction with the factorization in
Theorem~\ref{thm:c2_inner};
$Q$ is decoded as part of $A_\epsilon$, not supplied as shared
randomness. Expanding the mutual informations gives
\[
G_\epsilon=I(Q;Y,S)+(1-\epsilon)G_0+\epsilon G_1>0.
\]
The expected source distortion, sensing cost, and input cost
are the corresponding convex combinations of the branch
expectations. They may exceed the original budgets, but are
bounded by budgets converging to $(D_u,D_s,B)$ as
$\epsilon\downarrow0$, since all three costs are bounded.
Applying strict achievability to each mixed law and then
using the diagonal selection in
Appendix~\ref{app:c2_cc} establishes the target limsup
constraints. No positive information margin at the original
fixed budgets is assumed.

It remains to consider the case where $I(X;Y|S)=0$ for
every input distribution independent of $(U,S)$.
Choosing an input distribution with full support implies
$P_{Y|SX}=P_{Y|S}$ for every $s$ with $P_S(s)>0$.
Thus $Y-S-(U,A,X)$, while $A-U-S$ follows from
the factorization in
Theorem~\ref{thm:c2_inner}. Consequently,
\[
\begin{aligned}
0&\le I(U,X,Z;\widetilde V|A,Y,S)\\
&\le I(A;Y,S)-I(A;U)=-I(A;U|S)\le0.
\end{aligned}
\]
All terms vanish, and the chain rule gives
$I(U;A,\widetilde V,Y|S)=0$.
The reconstruction using $S$ alone defined in
Appendix~\ref{app:c2_cc} therefore has distortion no greater
than $\phi(A,\widetilde V,Y,S)$. The deterministic input
schedule described there realizes the given marginal
$P_{X|U}$ and its sensing and input costs. This proves
achievability also in this degenerate case.

\subsection{Inclusion of the First Achievable Region}

It remains to establish the claimed inclusion of the
achievable region in Theorem~\ref{thm:c2_wz_cc}.
Under \eqref{eq:c2_wz_cc_condition}, choose a
conditional pmf $P_{V|U}$ attaining $R_{U|S}(D_u)$ and a conditional pmf
$P_{A,X|U}$ certifying that condition, so that
$I(U;V|S)\le I(A;Y,S)-I(A;U)$, and form
\begin{equation}
P_{USAVXYZ}=P_{US}P_{V|U}P_{A,X|U}P_{YZ|SX}.
\label{eq:c2_wz_cc_law}
\end{equation}
Set $\widetilde V=V$ and use $\hat U=\phi(V,S)$,
ignoring $A,Y$ in the reconstruction. Then
\allowdisplaybreaks[0]
\begin{align*}
I(U,X,Z;V|A,Y,S)
&=I(U;V|S)-I(V;A,Y|S)\\
&\le I(U;V|S).
\end{align*}
The first equality follows from
$V-U-(A,X,Z,Y,S)$ under the displayed law.
The source distortion and both input constraints are unchanged.
Hence \eqref{eq:c2_inner_condition} holds, and the non-strict
achievability result proved above establishes the inclusion.
\par\vfill 

\section{Proof of Theorem~\ref{thm:c2_outer}}
\label{app:c2_outer}

Fix an arbitrary deterministic length-$n$ code.
Its joint distribution is
\begin{equation}
\begin{aligned}
&P(u^n,s^n,x^n,y^n,z^n)\\
&=\prod_{i=1}^n\Bigl[P_{US}(u_i,s_i)
  \mathbf1\{x_i=x_i(u^n,z^{i-1})\}\\
&\hspace{5em}\times P_{YZ|SX}(y_i,z_i|s_i,x_i)\Bigr].
\end{aligned}
\label{eq:c2_block_law}
\end{equation}
The conditional-distribution identity below is derived directly from
\eqref{eq:c2_block_law}.

\subsection{Auxiliary Variable and Factorization}

Given $U_i$, the current state is independent of
past states, past channel noise, and the other
source symbols. Therefore, $X_i-U_i-S_i$.

Next, fix $(U^n,S^{-i},X^n,Z^n,Y^{-i})$.
The only factor in \eqref{eq:c2_block_law}
that depends on $(s_i,y_i)$ is
\[
P_{S|U}(s_i|u_i)P_{YZ|SX}(y_i,z_i|s_i,x_i).
\]
All assertions about conditional laws are restricted to
positive-probability conditioning events. Although future
inputs may depend on $Z_i$, this variable is fixed under the
conditioning, so their encoding indicators do not depend
further on $(s_i,y_i)$.
Thus,
\begin{equation}
\begin{aligned}
&P(S_i,Y_i|U^n,S^{-i},X^n,Z^n,Y^{-i})\\
&\qquad=P(S_i,Y_i|U_i,X_i,Z_i).
\end{aligned}
\label{eq:c2_strong_posterior}
\end{equation}

Define
\begin{equation}
\overline V_i
=(U^{i-1},Y^{i-1},Y_{i+1}^n,S^{-i}).
\label{eq:c2_outer_aux}
\end{equation}
Average \eqref{eq:c2_strong_posterior} over the omitted
variables while holding $(U_i,X_i,Z_i,\overline V_i)$ fixed.
The right-hand side is unchanged. This proves, rather than
assuming from a reduced conditioning set,
\[
\overline V_i-(U_i,X_i,Z_i)-(Y_i,S_i).
\]
Moreover, the actual reconstruction $\hat U_i$
is a function of $(\overline V_i,Y_i,S_i)$.

Let $Q$ be independent and uniform on $[1:n]$, and set
\[
(U,S,X,Y,Z)=(U_Q,S_Q,X_Q,Y_Q,Z_Q),
\qquad
\overline V=(\overline V_Q,Q).
\]
The resulting distribution has the factorization specified in Theorem~\ref{thm:c2_outer}.
The source distortion and input cost equal their
time averages for the original code, while
\eqref{eq:c2_sensing_lower} gives the sensing constraint.

\subsection{First Information Constraint}

The i.i.d. source law and the chain rule give
the exact identity
\begin{align}
&\sum_i I(U_i;\overline V_i,Y_i|S_i)\nonumber\\
&=\sum_i\bigl[
H(U_i|S_i)
-H(U_i|U^{i-1},Y^n,S^n)\bigr]\nonumber\\
&=I(U^n;Y^n|S^n).
\label{eq:c2_outer_source_identity}
\end{align}
On the other hand,
\begin{align}
&I(U^n;Y^n|S^n)\nonumber\\
&=H(Y^n|S^n)
-\sum_iH(Y_i|U^n,S^n,Y^{i-1})\nonumber\\
&\le\sum_iH(Y_i|S_i)
-\sum_iH(Y_i|U^n,S^n,Y^{i-1},Z^{i-1})\nonumber\\
&=\sum_iI(X_i;Y_i|S_i).
\label{eq:c2_outer_channel_bound}
\end{align}
Since $Q$ is independent of $(U,S)$ and $Q-(X,S)-Y$, time-sharing gives
\begin{align*}
I(U;\overline V,Y|S)
&=\frac1n\sum_i I(U_i;\overline V_i,Y_i|S_i)\\
&\le I(X;Y|S,Q)\\
&\le I(X;Y|S).
\end{align*}
The last inequality follows from
$H(Y|S,Q)\le H(Y|S)$ and
$H(Y|X,S,Q)=H(Y|X,S)$.
Subtracting $I(U;Y|S)$ from both sides and using
$U-(X,S)-Y$, which implies
$I(X;Y|S)-I(U;Y|S)=I(X;Y|U,S)$, gives
\[
I(U;\overline V|Y,S)\le I(X;Y|U,S).
\]

\subsection{Second Information Constraint}

By the definition of the auxiliary,
\begin{align}
H(U,Z|\overline V,Y,S)
&=\frac1n\sum_iH(U_i,Z_i|U^{i-1},Y^n,S^n)\nonumber\\
&\ge\frac1n\sum_i
H(U_i,Z_i|U^{i-1},Z^{i-1},Y^n,S^n)\nonumber\\
&=\frac1nH(U^n,Z^n|Y^n,S^n).
\label{eq:c2_outer_entropy_lower}
\end{align}
Expanding first in $U^n$ and then in the time-ordered pairs
$(Y_i,Z_i)$, deterministic strictly causal encoding gives
$H(U^n|S^n)=nH(U|S)$ and
\[
H(Y_i,Z_i|U^n,S^n,Y^{i-1},Z^{i-1})
=H(Y_i,Z_i|X_i,S_i).
\]
Thus the memoryless channel gives
\begin{equation}
H(U^n,Y^n,Z^n|S^n)
=nH(U|S)+\sum_iH(Y_i,Z_i|X_i,S_i).
\end{equation}
Combining this identity with
$H(Y^n|S^n)\le nH(Y|S)$ and
\eqref{eq:c2_outer_entropy_lower} yields
\[
H(U,Z|\overline V,Y,S)
\ge H(U|S)+H(Y,Z|X,S)-H(Y|S).
\]
Consequently,
\begin{align}
I(U,Z;\overline V|Y,S)
&\le H(U,Z|Y,S)-H(U|S)\nonumber\\
&\quad-H(Y,Z|X,S)+H(Y|S)\nonumber\\
&=H(Y,Z|U,S)-H(Y,Z|X,S)\nonumber\\
&=I(X;Y,Z|U,S),
\end{align}
where the penultimate equality follows from
\[
H(U,Z|Y,S)+H(Y|S)-H(U|S)
=
H(Y,Z|U,S).
\]
This proves the second constraint.

\subsection{Cardinality and Asymptotic Constraints}

Fix $P_{USXYZ}$.
Under $\overline V-(U,X,Z)-(Y,S)$,
each value $\overline v$ is represented by a posterior
$P_{UXZ|\overline V=\overline v}$ and a
reconstruction function.

Preserving $P_{UXZ}$ requires
$|\mathcal U||\mathcal X||\mathcal Z|-1$ real quantities.
Additionally preserve
\[
H(U|\overline V,Y,S),\quad
H(U,Z|\overline V,Y,S),\quad
\mathbb E[d_u(U,\hat U)].
\]
Keep the channel-induced conditional law
$P_{YS|UXZ}$ fixed on the support of $P_{UXZ}$; it is
proportional to $P_{S|U}P_{YZ|SX}$. Assign it arbitrarily
off that support. The preserved posteriors are supported there,
so averaging them recovers the original channel law.
There are $|\mathcal U||\mathcal X||\mathcal Z|+2$ real
quantities in total. Caratheodory's theorem, applied to the
posterior together with its corresponding reconstruction function,
therefore requires at most one more support point and gives
{
\begin{equation}
|\overline{\mathcal V}|
\le|\mathcal U||\mathcal X||\mathcal Z|+3.
\label{eq:c2_gout_card}
\end{equation}
}
The input-cost and sensing-distortion constraints depend only
on the preserved marginal distribution.

Thus a common finite auxiliary alphabet suffices
for all blocklengths.
For an asymptotically achievable code sequence,
take a convergent subsequence of the reduced
distributions and reconstruction functions.
Compactness and continuity of entropy and bounded
costs give a limiting distribution satisfying
both information constraints at the target budgets.
This proves Theorem~\ref{thm:c2_outer}.

\section{Evaluation of the Correlated Source--State Model}
\label{app:c2_examples}

Under the correlated source--state model specified
in Section~\ref{sec:numerical_examples}, the state is
recovered exactly when $X=T$. Otherwise only $K$ is unknown,
with minimum error probability $1/2$.
Since
\[
\Pr\{T\ne x|U=u\}
=\begin{cases}
1/4,&x=u,\\
3/4,&x\ne u,
\end{cases}
\]
the sensing cost is
$c(u,x)=\frac18+\frac14\mathbf1\{x\ne u\}$.

Throughout this appendix, $(D_s,B)=(0.155,0.05)$ and $b(x)=x$.
Thus the two input constraints are
\begin{equation}
\Pr\{X\ne U\}\le\frac3{25},\qquad
\Pr\{X=1\}\le\frac1{20}.
\label{eq:c2_example_input_constraints}
\end{equation}
All entropy and mutual-information values below are in bits.

\subsection{SSCC Benchmark and Cover--Chiang-Based Rate}

For an SSCC benchmark input, write $t=P_X(1)$. Independence gives
$\Pr\{X\ne U\}=0.1+0.8t$, so $0\le t\le1/40$.
Since $Y=X$ and $X$ is independent of $S$,
\[
C_{\rm sep}(D_s,B)=\max_{0\le t\le1/40}H_b(t)=H_b(1/40).
\]
In particular, the separated input set is nonempty.

For any law with the factorization in
Theorem~\ref{thm:c2_wz_cc},
$(A,X)-U-S$ and $Y=X$ imply
\begin{align*}
I(A;X,S)-I(A;U)
&=I(A;X|U,S)-I(A;U|X,S)\\
&\le H(X|U,S)=H(X|U).
\end{align*}
For every feasible $P_{X|U}$, choosing $A=X$ attains this
upper bound. Thus maximizing the right-hand
side of Theorem~\ref{thm:c2_wz_cc} amounts to maximizing
$H(X|U)$ under
\eqref{eq:c2_example_input_constraints}.

To solve this maximization, parameterize the joint input law
by $t=P_X(1)$ and $w=P(U=1,X=1)$. These parameters exhaust
all binary input laws with the prescribed source marginal.
Its feasibility conditions are
\[
0\le t\le\frac1{20},\qquad
\max\left\{0,\frac{t-1/50}{2}\right\}\le w\le t,
\]
and
\[
H(X|U)=\frac9{10}H_b\!\left(\frac{t-w}{9/10}\right)
       +\frac1{10}H_b\!\left(\frac{w}{1/10}\right).
\]
For fixed $t$, this is concave in $w$, with its unconstrained
maximum at $w=t/10$, corresponding to $X$ independent of $U$.
For $t\le1/40$, that choice is feasible and yields $H_b(t)$.
For $1/40\le t\le1/20$, the maximizing choice is the lower
feasible endpoint $w=(t-1/50)/2$. At that endpoint,
\begin{align*}
P(X=1|U=0)&=\frac{t/2+1/100}{9/10},\\
P(X=1|U=1)&=5t-\frac1{10}.
\end{align*}
Both probabilities increase with $t$ and remain below $1/2$.
Consequently the objective increases up to $t=1/20$.
The global maximizer is therefore
\begin{equation}
P(X=1|U=0)=\frac7{180},\qquad
P(X=1|U=1)=\frac3{20},
\label{eq:c2_example_input_pmf}
\end{equation}
which proves \eqref{eq:c2_example_rcc}.
Both constraints in \eqref{eq:c2_example_input_constraints}
hold with equality under this $P_{X|U}$.

\subsection{Wyner--Ziv Function}

Since $K$ is independent of $(U,T)$ and of any $V$ generated
from $U$, it provides no improvement over the optimal
reconstruction based on $(V,T)$. In particular,
\[
H(U|S)=H(U|T)=H_b(0.1)+H_b(0.25)-H_b(0.3).
\]
For a posterior $q=P(U=1|V=v)$, define
\[
g(q):=H_b(q)+H_b(1/4)-H_b(1/4+q/2).
\]
Then $H(U|T,V=v)=g(q)$, and the minimum Hamming error
using $(v,T)$ is
\[
d(q)=\min\{q,1-q,1/4\}.
\]
Thus any feasible Wyner--Ziv description satisfies
$\mathbb E[d(q)]\le D_u$ and
\[
I(U;V|S)=H(U|S)-\mathbb E[g(q)].
\]
The function $g$ is symmetric about $1/2$ and strictly concave
on $(0,1)$. Indeed,
\[
g''(q)=
-\frac{3/16}{\ln(2)\,q(1-q)(1/4+q/2)(3/4-q/2)}<0.
\]
For $0<D_u\le1/20$, its tangent at $D_u$ gives the following bound for $0\le q\le1$:
\begin{equation}
g(q)\le g(D_u)+g'(D_u)\bigl(d(q)-D_u\bigr).
\label{eq:c2_example_wz_tangent}
\end{equation}
To verify this, concavity and symmetry suffice on
$[0,1/4]\cup[3/4,1]$. On the remaining interval,
$g(q)\le g(1/2)=H_b(1/4)$ and $d(q)=1/4$.
The excess of the tangent at $1/4$ over $H_b(1/4)$ decreases
with $D_u$, since its derivative is $(1/4-D_u)g''(D_u)<0$.
At $D_u=1/20$ this excess is
\[
\frac18\log_2\!\left(\frac{11^3\,29^5}{2^8\,19^6}\right)>0,
\]
which proves \eqref{eq:c2_example_wz_tangent} throughout the
claimed range. Since $g'(D_u)>0$, averaging this bound yields
$I(U;V|S)\ge H(U|S)-g(D_u)$.

Equality is attained by a binary $V$ with posterior values
$P(U=1|V=0)=D_u$ and $P(U=1|V=1)=1-D_u$, where
\[
P(V=1)=\frac{0.1-D_u}{1-2D_u}.
\]
These choices preserve $P_U$, induce a valid $P_{V|U}$,
and, with $\hat U=V$, have distortion $D_u$.
Generating $S$ conditionally on $U$ preserves $V-U-S$.
The case $D_u=0$ follows by taking $V=U$.
This establishes \eqref{eq:c2_example_wz} on $[0,0.05]$.

\subsection{Feedback-Assisted Description and Strict Improvement}

Use \eqref{eq:c2_example_input_pmf}, set $A=X$, and note that
$T$ is always available from $Z$. The probabilities of
$(X,T)=(0,0),(0,1),(1,0),(1,1)$ are, respectively,
\[
\frac{67}{100},\quad\frac7{25},\quad
\frac3{100},\quad\frac1{50}.
\]
For convenience, denote the associated source posteriors by
$p_{xt}:=P(U=1|X=x,T=t)$. They are
\[
p_{00}=\frac{17}{536},\quad
p_{01}=\frac{51}{224},\quad
p_{10}=\frac18,\quad
p_{11}=\frac9{16}.
\]
For this fixed input law, binary Hamming reconstruction gives the
description-rate lower bound $H(U|X,T)-H_b(D_u)$, which the
following construction attains.
For $0\le D_u\le17/536$, define
\[
\theta_{xt}:=\frac{p_{xt}-D_u}{1-2D_u}.
\]
Each $\theta_{xt}$ belongs to $[0,1]$. Specify a binary
supplementary description by the backward conditional law
\[
\begin{aligned}
P(\widetilde V=1|x,t)&=\theta_{xt},\\
P(U\ne\widetilde V|\widetilde V,x,t)&=D_u.
\end{aligned}
\]
This has the required marginal $P_{U|X,T}$. Equivalently,
the encoder can implement the forward conditional pmf $P_{\widetilde V|U,X,T}$
\begin{align*}
P(\widetilde V=1|U=1,x,t)
&=\frac{(1-D_u)\theta_{xt}}{p_{xt}},\\
P(\widetilde V=1|U=0,x,t)
&=\frac{D_u\theta_{xt}}{1-p_{xt}}.
\end{align*}
Because $T$ is a function of $Z$, this is an admissible
$P_{\widetilde V|U,A,X,Z}$ in the joint distribution of
Theorem~\ref{thm:c2_inner}.
Set $\hat U=\widetilde V$. Its source distortion is $D_u$,
and $H(U|\widetilde V,X,T)=H_b(D_u)$.
The conditional pmf does not depend on $K$, and $Z$ is deterministic given $(X,S)$.
Consequently,
\begin{align}
I(U,X,Z;\widetilde V|A,Y,S)
&=I(U;\widetilde V|X,T)\nonumber\\
&=H(U|X,T)-H_b(D_u),
\label{eq:c2_example_description_rate}\\
I(A;Y,S)-I(A;U)&=H(X|U).
\label{eq:c2_example_channel_rate}
\end{align}
The conditional pmf $P_{X|U}$ also gives
\begin{align}
H(X|T)
&=\frac7{10}H_b\!\left(\frac3{70}\right)
 +\frac3{10}H_b\!\left(\frac1{15}\right)\nonumber\\
&\approx0.2846772797,
\label{eq:c2_example_output_entropy}\\
H(U|X,T)
&=H(U|S)+H(X|U)-H(X|T)\nonumber\\
&\approx0.3887465171.
\label{eq:c2_example_residual_entropy}
\end{align}
At $D_u=0.018$,
\eqref{eq:c2_example_description_rate} and
\eqref{eq:c2_example_channel_rate} are approximately
$0.2586876709$ and $0.2744409780$, respectively.
Theorem~\ref{thm:c2_inner} therefore applies, whereas
$R_{U|S}(0.018)\approx0.2828794547$ exceeds the supremum of
the achievable message rates in \eqref{eq:c2_example_rcc}
and the capacity of the SSCC benchmark.
This verifies the strict case of
\eqref{eq:c2_example_condition_t3}. At $D_u=0.03$,
$R_{U|S}(D_u)\approx0.2275109324$ lies strictly between
$C_{\rm sep}$ and the supremum of achievable message rates
in \eqref{eq:c2_example_rcc},
witnessing a strict gain of Theorem~\ref{thm:c2_wz_cc}
over the SSCC benchmark.

\subsection{Evaluation and Tightness of the Converse}

Since $Y=X$ and $Z$ is deterministic given $(X,S)$, the two
necessary conditions in Theorem~\ref{thm:c2_outer} both reduce to
\[
I(U;\overline V|X,S)\le H(X|U,S).
\]
For $D_u\le1/2$, binary Hamming reconstruction gives
$H(U|\overline V,X,S)\le H_b(D_u)$. Hence
\begin{align}
H(U|S)-H_b(D_u)
&\le I(U;X,\overline V|S)\nonumber\\
&\le I(U;X|S)+H(X|U,S)\nonumber\\
&=H(X|S)=H(X|T).
\label{eq:c2_example_converse_entropy}
\end{align}
The last equality uses the input factorization and the
independence of $K$ from $(U,T)$.

It remains to maximize $H(X|T)$ over all feasible input laws,
not merely the $P_{X|U}$ used for achievability. With the same
parameters $t=P_X(1)$ and $w=P(U=1,X=1)$ as above,
\[
\begin{aligned}
H(X|T)
={}&\frac7{10}H_b\!\left(\frac{3t/4-w/2}{7/10}\right)\\
&+\frac3{10}H_b\!\left(\frac{t/4+w/2}{3/10}\right).
\end{aligned}
\]
For fixed $t$, concavity gives the unconstrained maximum at
$w=t/10$. This is feasible for $t\le1/40$, yielding $H_b(t)$.
For $1/40\le t\le1/20$, the maximum is at
$w=(t-1/50)/2$. The two conditional input probabilities there
are $(t/2+1/200)/(7/10)$ and $(t/2-1/200)/(3/10)$;
both increase with $t$ and remain below $1/2$.
Thus the global maximum is attained at $t=1/20$ and equals
\eqref{eq:c2_example_output_entropy}.
Consequently, every feasible $P_{X|U}$ satisfies
$H(X|T)\le\frac7{10}H_b(3/70)+\frac3{10}H_b(1/15)$,
with equality under \eqref{eq:c2_example_input_pmf}.
By Theorem~\ref{thm:c2_outer}, every achievable tuple must satisfy
\eqref{eq:c2_example_converse_entropy} for some feasible input law.
Substituting in \eqref{eq:c2_example_converse_entropy} proves
$D_u\ge D_u^\star$ with $D_u^\star$ given by
\eqref{eq:c2_example_optimal_distortion}. Values $D_u>1/2$
trivially satisfy this lower bound.

Conversely, under \eqref{eq:c2_example_input_pmf}, subtracting
\eqref{eq:c2_example_description_rate} from
\eqref{eq:c2_example_channel_rate} gives
\[
H(X|T)-H(U|S)+H_b(D_u).
\]
It is nonnegative whenever
$D_u^\star\le D_u\le17/536$.
Theorem~\ref{thm:c2_inner} therefore establishes achievability
throughout this interval, including the boundary point (see
\appref{E}{app:c2_inner}), at the
fixed sensing and cost budgets. Larger distortion
budgets are achievable by retaining the same code sequence.
Thus $D_u^\star$ is the operational minimum source distortion
at the stated budgets.
As a direct check, at the boundary the same $P_{X|U}$ with
$\overline V=(X,\widetilde V)$ satisfies both converse
inequalities with equality and the same reconstruction.
This tightness statement concerns the specified budget pair;
it does not establish a general matching characterization.

}

\newpage
\bibliographystyle{IEEEtran}
\bibliography{refer}

\end{document}